\documentclass[11pt]{article}
\usepackage[T1]{fontenc}
\usepackage{newtxtext}
\usepackage{lmodern}
\usepackage{textcomp}
\def\nno{\nonumber}

\usepackage[top=0.9in, bottom=0.9in, left=0.7in, right=0.7in]{geometry}
\usepackage{epsfig, color}
\usepackage{color,graphicx}
\usepackage{amsfonts}
\usepackage{color}
\usepackage{amsmath}
\usepackage{amssymb}
\usepackage{amsthm}
\usepackage{appendix}
\usepackage{fancybox}
\usepackage{ulem}
\usepackage{subfigure}
\usepackage{multirow}
\usepackage{bbm}
\usepackage{enumerate}
\usepackage{rotating}
\usepackage{float}
\usepackage{epstopdf}
\usepackage{tcolorbox}
\usepackage{mdframed}
\usepackage{framed}
\usepackage{graphbox}
\usepackage{makecell}
\usepackage{diagbox}
\usepackage{booktabs}
\usepackage{threeparttable}
\usepackage{appendix}
\usepackage[colorlinks=true, citecolor=blue]{hyperref}
\usepackage{longtable}
\usepackage{natbib}
\usepackage{url}
\usepackage{tabularx}
\usepackage{subcaption}
\usepackage{subfigure}
\usepackage{caption}
\usepackage{relsize}

\allowdisplaybreaks[4]

\newtheorem{theorem}{\hspace{6mm}Theorem}[section]

\newtheorem{remark}{\hspace{6mm}Remark}[section]
\newtheorem{example}{\textsc{Example}}[section]
\newtheorem{proposition}{\hspace{6mm}Proposition}[section]

\def\cblue{\textcolor{blue}}
\def\cpurple{\textcolor{purple}}
\def\bs{\boldsymbol}
\def\nno{\nonumber}
\newcommand{\shuffle}{\mathbin{\text{\scriptsize $\sqcup\!\sqcup$}}}

\title{Option pricing under non-Markovian stochastic volatility models: A deep signature approach}
\author{Jingtang Ma\thanks{\rm School of Mathematics, Big Data Laboratory on Financial Security and Behavior (Laboratory of Philosophy and Social Sciences, Ministry of Education),
Southwestern University of Finance and Economics, Chengdu, 611130, PR China. Email: mjt@swufe.edu.cn}  
\and Xianglin Wu\thanks{\rm School of Mathematics, Southwestern University of Finance and Economics, Chengdu, 611130, PR China. Email: 1210202Z1003@smail.swufe.edu.cn}
\and Wenyuan Li\thanks{\rm Corresponding author. Department of Statistics and Actuarial Science, The University of Hong Kong, Pokfulam, Hong Kong, PR China.  Email: wylsaas@hku.hk.} }

\begin{document}

\maketitle

\begin{abstract}
This paper studies the option pricing problem in which the underlying asset follows a non-Markovian stochastic volatility model. Classical partial differential equation methods face significant challenges in this context, as the option prices depend not only on the current state, but also on the entire historical path of the process. To overcome these difficulties, we reformulate the asset dynamics as a rough stochastic differential equation and then represent the rough paths via linear or non-linear combinations of time-extended Brownian motion signatures. \cpurple{This representation transforms a rough stochastic differential equation to a classical stochastic differential equation, allowing us to characterize the value function as a random field and apply the It\^o-Wentzell formula to derive the corresponding pricing PDE. }We propose a deep signature approach for both linear and nonlinear representations and rigorously prove the convergence of the algorithm. Numerical examples demonstrate the effectiveness of our approach for both Markovian and non-Markovian volatility models, offering a theoretically grounded and computationally efficient framework for option pricing.
\end{abstract}
	
	\section{Introduction}
	The pricing problem lies at the heart of mathematical finance, which is equivalent to computing an expectation such as
	\begin{equation}
	\label{eq-intro-pricing}
	\mathbb{E}[\Phi(S_T)],
	\end{equation}
	where $\Phi$ is a sufficiently regular payoff function and $T$ denotes maturity. A typical choice of the underlying asset $S$ is the stochastic volatility model as follows
	\begin{equation}
	\label{eq-intro-asset}
	dS_t=f(t,S_t)v_tdW_t+g(t,S_t)v_tdB_t, \quad S_0=s_0, \; 0< t\leq T,
	\end{equation}
	where $(f,\,g)$ is a pair of deterministic functions and $W$ and $B$ are independent Brownian motions in a complete filtered probability space $(\Omega,\mathcal{F}_t,P)$. Here, the randomness in \eqref{eq-intro-asset} only arises from the two independent Brownian motions, implying that the filtration $\mathcal{F}_t$ is generated by $\mathcal{F}^W_t$ and $\mathcal{F}^B_t$, the filtration induced by $W$ and $B$, respectively, i.e. $\mathcal{F}_t=\mathcal{F}^W_t\lor\mathcal{F}^B_t$.
	As for the volatility process $v_t:0\leq t\leq T$, we assume that it is $\mathcal{F}^W_t$-progressive with bounded sample paths throughout this work. This assumption is very general, as it includes nearly all classical stochastic volatility processes in use, as well as the non-Markovian cases, such as rough volatility processes (see Table~\ref{tab-example-vol-processes}).
	The formulation of the volatility process $v$ plays an important role in pricing applications. Specifically, classical models, which are typically Markovian, allow the Feynman-Kac theorem to elegantly connect the pricing problem \eqref{eq-intro-pricing} with the corresponding partial differential equations (PDEs).
	However, this framework faces significant challenges when applied to more complex volatility processes that lack the Markovian and semimartingale structure.
	To overcome these difficulties, we integrate rough path theory with signature transforms to construct a Markovian representation equivalent to~\eqref{eq-intro-asset}, thereby enabling the application of classical mathematical finance tools.
	
	Specifically, we first reformulate the dynamic~\eqref{eq-intro-asset} to a rough stochastic differential equation (RSDE) with the It\^o rough path lift (see \cite{lyons2007differential} and \cite{bank2025rough}) as follows
	\begin{equation}
		\label{eq-intro-RSDE}
        dX_t=f(t,X_t)d\mathbf{I}_t+g(t,X_t)\mathbf{v}_tdB_t,\quad X_0=s_0,\; 0<t\leq T,
	\end{equation}
	where $\mathbf{I}_t=(I_t,\mathbb{I}_t)$ is the It\^o rough path lift of the fixed realization integrator $I_t:=\int_0^tv_sdW_s$ and $\mathbf{v}_t$ is defined by $\mathbf{I}_t$. 
	\cpurple{
	The deterministic nature of the rough pair $(\mathbf{v},\mathbf{I})$ implies that the RSDE~\eqref{eq-intro-RSDE} admits a pathwise solution.  
	Next, we express the rough pair $(\mathbf{v},\mathbf{I})$ in~\eqref{eq-intro-RSDE} as infinite linear or nonlinear combinations of the time-extended Brownian motion's path signature $(\widehat{\mathbb{W}}_t^{\infty})_{t\in[0, T]}$. 
	These signature representations allow us to rewrite the RSDE~\eqref{eq-intro-RSDE} as a classical, Markovian stochastic differential equation (SDE) driven by $B$ (see this Markovization technique in \cite{lyons2007differential} and \cite{bank2025rough}).}
	For practical purposes, we further truncate the infinite signature representation to the finite case and prove the convergence of the truncation representation. Lastly, we provide extensive examples to demonstrate the algorithm's effectiveness and develop a deep signature method that is easily extended to general volatility processes.
	
	The non-Markovian stochastic volatility model lies in the field of rough path theory. \cite{lyons1998differential} first develops the theoretical foundations for the deterministic rough differential equations (RDEs). \cite{diehl2017stochastic} extend the rough path framework to stochastic partial differential equations (SPDEs) with probabilistic representations. \cite{friz2021rough} introduce the RSDEs to provide a unified approach to SDEs and RDEs and prove the existence and uniqueness of its solution.  For more recent work, we refer to \cite{friz2024controlled}, \cite{bugini2024parameter}, \cite{bugini2024malliavin}, and \cite{li2024reflected}. In financial applications, \cite{bank2025rough} pioneered the use of rough partial differential equations (RPDEs) for option pricing by representing rough paths as weakly geometric ones. These paths were then approximated by  Lipschitz continuous paths, although the approximation error remains unquantified. 
	
	In this paper, we use a signature transform to study the non-Markovian volatility path. It is a convolution transform capturing the path information up to the current time (see \cite{chen1957integration}). By its universal linearization property together with machine learning techniques, the signature transform enables us to extract realistic path features \cpurple{from complex} volatility models (see, e.g.,~\cite{jaber2025hedging},~\cite{kidger2019deep},~\cite{tong2023sigformer} and~\cite{bayer2025pricing}). While these methods are easy to use, they lack mathematical tractability, precluding the use of classical mathematical finance tools. \cite{cuchiero2023signature} rigorously exploit the polynomial nature of generic primary processes and derive a tractable option pricing formula. \cite{jaber2024path} and~\cite{abi2025signature} further provide signature linear representations for linear path-dependent models with analytical coefficients.
	
	Our transformation of the dynamic process~\eqref{eq-intro-asset} into an RSDE is inspired by the state-of-the-art work \cite{bank2025rough}. 
    \cpurple{Building on this, we further characterize the value function of the European option as a random field. 
    Applying the It\^o-Wentzell formula then yields the classical pricing PDE. }
    However, our approach differs significantly in the treatment of rough paths within the RSDE framework. Specifically, we represent the rough path exactly through an infinite series of signatures, in contrast to their approach, which approximates the rough path as a weakly geometric rough path with Lipschitz continuous approximations. Our method eliminates the need for such approximations, offering a more precise characterization of the rough path dynamics. In addition, the linear signature representation technique was initially introduced in~\cite{abi2025signature} and~\cite{jaber2024path}, but the analytical forms are limited to linear path-dependent models, such as linear Volterra equations, linear delay equations, and Gaussian Volterra processes. Extending this approach to non-linear path-dependent models, such as the rough Heston (rHeston) model, remains an open challenge. To address this, we propose two novel \cpurple{numerical} approaches: (i) a deep linear signature approach, which employs neural networks to learn the time-dependent coefficients of a linear signature representation, and (ii) a deep nonlinear signature approach, which directly uses the signature as input to a neural network to capture non-linear dynamics. Numerical experiments demonstrate that the deep nonlinear signature approach significantly outperforms the linear method in terms of accuracy. Furthermore, we provide rigorous convergence proofs for both methods.
	
	\cpurple{
	Our paper contributes to the existing literature by incorporating the signature transform into the RSDE.
	This combined approach converts a non-Markovian system into an equivalent Markovian framework that can be solved pathwise.
    Within this framework, we can naturally characterize the value function as a random field and apply the It\^o-Wentzell formula, ultimately deriving the classical pricing PDE for European options.
	}
	Our contributions can be summarized in three aspects:
	\begin{itemize}
		\item This paper extends classical PDE tools, such as the Feynman-Kac theorem, to address option pricing problems in non-Markovian frameworks by leveraging \cpurple{signature transforms, rough path lifts, and the It\^o-Wentzell formula}.
		This approach bridges the gap between traditional stochastic calculus and path-dependent settings, enabling the analysis of complex financial models where standard Markovian assumptions fail.
		\item We propose two innovative neural network-based approaches for volatility process approximation: a deep linear signature approach and a deep nonlinear signature approach. These approaches enrich the literature by solving the general non-linear path-dependent models.
		\item We provide a comprehensive convergence analysis that incorporates truncation errors of the signature, approximation errors of the signature representation, and the overall pricing error. This demonstrates the reliability of our method.
	\end{itemize}
	
	The rest of the paper is organized as follows. In Section~\ref{sec-problem}, we formulate the problem by first converting the underlying asset price process into an equivalent RSDE, then connecting it to a classical SDE through signature representations. 
	\cpurple{
	Finally, the corresponding pricing PDE is derived using the It\^o-Wentzell formula. 
	}
	Rigorous convergence proofs are provided for both linear and non-linear signature representations. 
	In Section~\ref{sec-examples}, we compute the approximation errors of signature representation and European option pricing errors for both Markovian and non-Markovian models. Section~\ref{sec-conclusions} concludes.
	
	\section{Representing RSDE with finite linear and nonlinear combinations of signatures}
	\label{sec-problem}
		Let $(W_t)_{t\in[0,T]}$ and $(B_t)_{t\in[0,T]}$ be two independent standard Brownian motions.
		Consider an $(\mathcal{F}_t^W)$-progressive volatility process $(v_t)_{t\in[0,T]}$ with bounded sample paths, and define the associated integral process $I_t:=\int_0^tv_sdW_s,\,t\in[0,T]$.
		We are interested in the following model
		\begin{equation}
		\label{eq-sv}
			dS_t=f(t,S_t)dI_t+g(t,S_t)v_tdB_t,\quad S_0=s_0,\;t\in[0,T],
		\end{equation}
		\cpurple{where $f,g\in C^{0,3}([0,T]\times\mathbb{R})$ are locally Lipschitz continuous and satisfy the global linear growth conditions on $[0,T]$,} and $s_0\in\mathbb{R}^+$ is the \cpurple{finite} initial value at $t=0$.
		The SDE~\eqref{eq-sv} is very general, as its volatility process $v$ encompasses all $(\mathcal{F}_t^W)$-progressive and continuous specifications, including the Ornstein-Uhlenbeck (OU) process, the mean-reverting geometric Brownian motion (mGBM), and the recent rough volatility variants (see, e.g., Table~\ref{tab-example-vol-processes}).
		
		Let $(S_t)_{t\in[0,T]}$ denote the underlying asset price process, and we consider the European option pricing problem as follows
		\begin{equation}
			\label{eq-pricing}
			\cpurple{\mathbb{E}\left[\Phi\left(S_T\right)\right],}
		\end{equation}
		where $\Phi$ is a Lipschitz continuous payoff function and $T$ denotes the maturity.
		In this paper, we focus on solving the pricing problem when the pair $(v,I)$ follows non-Markovian dynamics. In a classical Markovian setting, one can employ the Feynman-Kac theorem to establish an equivalence between the option pricing problem~\eqref{eq-pricing} and the solution to the PDE associated with~\eqref{eq-sv}.
		However, the non-Markovian nature, arising from the pair $(v,I)$, invalidates the well-posedness of the PDE solution derived via the Feynman-Kac theorem.
		To establish the Markovian property for the SDE~\eqref{eq-sv}, we begin by introducing the It\^o rough path lift of the martingale $I$ and its bracket as follows
		\[
		\mathbf{I}_t=(I_t,\mathbb{I}_t)=\left(I_t,\iint\limits_{0<s_1<s_2<t}dI_{s_1}dI_{s_2}\right),\quad
			[\mathbf{I}]_t=[I]_t=\int_0^tv^2_sds,\quad t\in[0,T].
		\]
		Moreover, we give the definition of $\mathbf{v}$ for almost every $t\in[0,T]$ as follows
		\begin{equation*}
			\left([\mathbf{I}]_t,\frac{d[\mathbf{I}]_t}{dt},\sqrt{\frac{d[\mathbf{I}]_t}{dt}}\right):=\left(\int_0^t\mathbf{v}_s^2ds,\mathbf{v}_t^2,|\mathbf{v}_t|\right)=\left(\int_0^tv_s^2ds,v^2_t,|v_t|\right).
		\end{equation*}
		{\color{purple}
		Following the existing literature \cite{bank2025rough}, we treat the rough path lift pair $(\mathbf{v},\mathbf{I})$ as fixed and deterministic. Crucially, this pathwise perspective motivates us to use the rough path framework: it allows us to cast the non-Markovian component of the underlying asset price process \eqref{eq-sv} into a deterministic setting, thereby bypassing the probabilistic analytical difficulties that typically arise from its non-Markovian property. 

		Replacing the pair $(v,I)$ in the SDE~\eqref{eq-sv} with $(\mathbf{v},\mathbf{I})$ yields the RSDE as follows
		\begin{equation}
			\label{eq-RSDE}	
			dX_t=f(t,X_t)d\mathbf{I}_t+g(t,X_t)\mathbf{v}_tdB_t,\quad X_0=s_0,\; t\in[0,T].
		\end{equation}
		For any fixed realization $\omega$ of $W$, the solution $S(\omega)$ to the SDE~\eqref{eq-sv} conditional on $\mathcal{F}^W_t$ and the solution $X(\omega)$ to the RSDE~\eqref{eq-RSDE} have the same distribution a.s. (see  Theorem 3.2 in~\cite{bank2025rough}).
		That is, for $t\in[0,T]$, we have
		\begin{equation*}
			\text{Law}(S_t|\mathcal{F}^W_t)(\omega)
			=\text{Law}(X_t)|_{\mathbf{I}_t=(I_t,\mathbb{I}_t)(\omega)},\quad \text{a.s. }\omega\in\Omega,
		\end{equation*}
		where the ``Law'' means the distribution. 
        In this case, $X(\omega)$ is a time-inhomogeneous Markovian process since it can be regarded as $S(\omega)$ conditional on $\mathcal{F}^W_t$ in a distributional sense. Specifically, non-Markovian pair $(v_t,I_t)$ is $\mathcal{F}^W_t$ measurable and $S_t$ is Markovian with $\mathcal{F}^B_t$, thus $S_t|\mathcal{F}^W_t$ is Markovian as all information of $(v_t,I_t), t \in [0,T]$ is given. As a consequence, we degenerate the non-Markovian pricing problem to a Markovian problem, and the option price~\eqref{eq-pricing} can be computed via the RSDE~\eqref{eq-RSDE} as
        \begin{equation*}
            \mathbb{E}\big[\Phi(S_T)\big] = \mathbb{E}\big[\mathbb{E}\big[\Phi(S_T)\big|\mathcal{F}^W_T\big]\big]
            =\mathbb{E}\big[\mathbb{E}\big[\Phi(X_T)\big|\mathbf{I}_T\big]\big]=\mathbb{E}\big[\Phi(X_T)\big].
        \end{equation*}}
		
        \subsection{Framework for linear signature representation}
		This section reformulates the RSDE\\
        \eqref{eq-RSDE} into a more tractable form to compute $\mathbb{E}\left[\Phi\left(X_T\right)\right]$.
		To this end, we introduce the path signature framework.
		Specifically, let
		$$\widehat{\mathbb{W}}_t^{\infty}(\omega) :=\left(\widehat{\mathbb{W}}_t^{(0)}(\omega),\widehat{\mathbb{W}}_t^{(1)}(\omega), \ldots,\widehat{\mathbb{W}}_t^{(n)}(\omega),\ldots\right),\quad t\in[0,T]$$
		denote the path signature of the time-extended Brownian motion $\widehat{W}_t(\omega):=(t,W_t(\omega)),\,t\in[0,T]$ in the Stratonovich sense (see \cite{cuchiero2023signature} for details). 
		{\color{purple}
		Note that $\widehat{\mathbb{W}}_t^{\infty}(\omega)$ and $W_t(\omega)$ share the same sample path $\omega$, and we will omit the explicit dependence on $\omega$ in what follows for notional simplicity. 
		We assume that the stochastic pair $(v,I)$ can be equivalently rewritten as infinite linear combinations of stochastic process $\widehat{\mathbb{W}}^{\infty}$ with the time-dependent and smooth coefficients $\bs{\ell}_t:=(\bs{\ell}^{(0)}_t,\bs{\ell}^{(1)}_t,\ldots,\bs{\ell}^{(n)}_t,\ldots)$ and $\bs{p}_t:=(\bs{p}^{(0)}_t,\bs{p}^{(1)}_t,\ldots,\bs{p}^{(n)}_t,\ldots)$, i.e., $v_t=\langle\bs{\ell}_t,\widehat{\mathbb{W}}_t^{\infty}\rangle$ and $I_t=\langle\bs{p}_t,\widehat{\mathbb{W}}_t^{\infty}\rangle$ for $t\in[0,T]$ (see~\cite{abi2025signature} and~\cite{jaber2024path} for details). 
		Moreover, when evaluated with a fixed sample $\omega$, the It\^o rough path lift $(\mathbf{v},\mathbf{I})$ is entirely pathwise and deterministic. 
		Therefore, its signature representation with the same coefficients $(\bs{\ell}_t,\bs{p}_t)$ is also pathwise and deterministic, corresponding directly to that fixed sample $\omega$. 
		}
		The time-dependent coefficients \cpurple{belong to} an extended tensor algebra defined within the admissible set $\mathcal{A}(\widehat{\mathbb{W}}^{\infty})$, which ensures that the infinite series remains well-defined:
		\cpurple{
		\begin{equation*}
			\mathcal{A}(\widehat{\mathbb{W}}^{\infty}):=\left\{\cdot\;:\sum\limits_{n=0}^{\infty}\left|\langle\cdot,\widehat{\mathbb{W}}_t^{(n)}\rangle\right|<\infty\text{ for all }t\in[0,T]\text{ a.s.}\right\}.
		\end{equation*}
		Applying the linear signature representations above, the RSDE~\eqref{eq-RSDE} can be rewritten as follows
		\begin{flalign}
			\label{eq-RSDE-rep-t}
			dX_t
			=&f(t,X_t)d\mathbf{I}_t+g(t,X_t)\mathbf{v}_tdB_t\\
			=&f(t,X_t)d\langle\bs{p}_t,\widehat{\mathbb{W}}_t^{\infty}\rangle+g(t,X_t)\langle\bs{\ell}_t,\widehat{\mathbb{W}}_t^{\infty}\rangle dB_t\nno\\
			=&f(t,X_t)\big(\langle\dot{\bs{p}}_t,\widehat{\mathbb{W}}_t^{\infty}\rangle dt+\langle\bs{p}_t,d\widehat{\mathbb{W}}_t^{\infty}\rangle\big)+g(t,X_t)\langle\bs{\ell}_t,\widehat{\mathbb{W}}_t^{\infty}\rangle dB_t\nno\\
			=&f(t,X_t)\big[\langle\dot{\bs{p}}_t,\widehat{\mathbb{W}}_t^{\infty}\rangle dt+\langle\bs{p}_t,\widehat{\mathbb{W}}_t^{\infty}\otimes(e_1+\frac{1}{2}e_2\otimes e_2)\rangle dt+\langle\bs{p}_t,\widehat{\mathbb{W}}_t^{\infty}\otimes e_2\rangle dW_t\big]\nno\\
			&+g(t,X_t)\langle\bs{\ell}_t,\widehat{\mathbb{W}}_t^{\infty}\rangle dB_t\nno\\
			=&\langle f(t,X_t)(\dot{\bs{p}}_t+\mathcal{D}_1\bs{p}_t+\frac{1}{2}\mathcal{D}_{22}^2\bs{p}_t),\widehat{\mathbb{W}}_t^{\infty}\rangle dt+\langle f(t,X_t)\mathcal{D}_2\bs{p}_t,\widehat{\mathbb{W}}_t^{\infty}\rangle dW_t\nno\\
			&+\langle g(t,X_t)\bs{\ell}_t,\widehat{\mathbb{W}}_t^{\infty}\rangle dB_t\nno\\
            =&\langle\bs{q}^f_t,\widehat{\mathbb{W}}_t^{\infty}\rangle dt+\langle f(t,X_t)\mathcal{D}_2\bs{p}_t,\widehat{\mathbb{W}}_t^{\infty}\rangle dW_t+\langle g(t,X_t)\bs{\ell}_t,\widehat{\mathbb{W}}_t^{\infty}\rangle dB_t,\quad t\in[0,T],\nno
		\end{flalign}
		where $\{e_1, e_2\}$ denotes the canonical basis of $\mathbb{R}^2$, and $\mathcal{D}$ acts as the adjoint operator of right-concatenation\footnote{\cpurple{Specifically, it is defined such that $\langle\mathcal{D}_i\bs{p}_t,\widehat{\mathbb{W}}_t^{\infty}\rangle:=\langle\bs{p}_t,\widehat{\mathbb{W}}_t^{\infty}\otimes e_i\rangle$ and $\langle\mathcal{D}_{ij}^2\bs{p}_t,\widehat{\mathbb{W}}_t^{\infty}\rangle:=\langle\bs{p}_t,\widehat{\mathbb{W}}_t^{\infty}\otimes e_i\otimes e_j\rangle$ for $i,j\in\{1,2\}$. }}~in the dual space of the signature algebra (see Section 5.1, \cite{lemercier2024log}). 
		The coefficient $\bs{q}^f_t$ reads as
		\begin{equation*}
			\label{eq-coefficient-q-t}
			\bs{q}^f_t:=f(t,X_t)(\dot{\bs{p}}_t+\mathcal{D}_1\bs{p}_t+\frac{1}{2}\mathcal{D}_{22}^2\bs{p}_t).
		\end{equation*}
		}
		
		\cpurple{
		Notably, we can derive numerically equivalent time-independent coefficients $\bs{\ell}$ and $\bs{p}$ for the linear signature representations of specific $(\mathbf{v},\mathbf{I})$ through Taylor expansion (see~\cite{abi2025signature}).
		Building on this result, we present the Proposition~\ref{proposition-time-independent} to reconstruct the RSDE~\eqref{eq-RSDE} using the time-independent signature representation coefficients $\bs{\ell}$ and $\bs{p}$, namely $\mathbf{v}_t=\langle\bs{\ell},\widehat{\mathbb{W}}_t^{\infty}\rangle$ and $\mathbf{I}_t=\langle\bs{p},\widehat{\mathbb{W}}_t^{\infty}\rangle$ for $t\in[0,T]$.
		}
		\begin{proposition}[Linear signature representation with time-independent coefficients]
			\label{proposition-time-independent}
			\cpurple{
			Let $\widehat{W}_t:=(t,W_t)$ be a time extended one-dimensional Brownian motion, and $\widehat{\mathbb{W}}^{\infty}$ be the infinite sequence of its path signature.
			Suppose the pair $(\mathbf{v}_t, \mathbf{I}_t)$ can be expressed as $\mathbf{v}_t=\langle\bs{\ell},\widehat{\mathbb{W}}_t^{\infty}\rangle$ and $\mathbf{I}_t=\langle\bs{p},\widehat{\mathbb{W}}_t^{\infty}\rangle$, where $\bs{\ell}= (\bs{\ell}^{(0)},\bs{\ell}^{(1)},\bs{\ell}^{(2)},\ldots)$ and $\bs{p}= (\bs{p}^{(0)},\bs{p}^{(1)},\bs{p}^{(2)},\ldots)$ are time-independent tensors in the dual space of the signature algebra that also belong to the admissible set $\mathcal{A}(\widehat{\mathbb{W}}^{\infty})$. 
			Then the RSDE~\eqref{eq-RSDE} can be represented as follows
			\begin{equation*}
				\label{eq-RSDE-rep}
				d\widetilde{X}_t
				=\langle\widetilde{\bs{q}}^f_t,\widehat{\mathbb{W}}_t^{\infty}\rangle dt+\langle f(t,X_t)\mathcal{D}_2\bs{p},\widehat{\mathbb{W}}_t^{\infty}\rangle dW_t+\langle g(t,X_t)\bs{\ell},\widehat{\mathbb{W}}_t^{\infty}\rangle dB_t,\quad t\in[0,T],
			\end{equation*}
			where 		
			\begin{equation*}
				\widetilde{\bs{q}}^f_t:=f(t,X_t)(\mathcal{D}_1\bs{p}+\frac{1}{2}\mathcal{D}_{22}^2\bs{p}).
			\end{equation*}
			}
		\end{proposition}
        \begin{proof}
        	\cpurple{
            The result follows directly from Equation \eqref{eq-RSDE-rep-t}.
        	}
        \end{proof}
		
		{\color{purple}
		Since the non-Markovian component of the underlying asset price process has been replaced by deterministic rough paths and subsequently by signature representations, the process $X(\omega)$ becomes a time-inhomogeneous Markov process. 
		Consequently, when characterizing the European option price $\mathbb{E}\big[\Phi(X_T)\mid\widehat{\mathbb{W}}^{\infty}_T,X_t=x\big]$, the associated value function $u(t,x)$ inherently depends on a path $\omega$. 
		Assume that $u(t,x)$ is given as follows:
		\begin{equation}
			\label{eq-u-field}
			du(t,x)=\Psi_1(t,x)dt+\Psi_2(t,x)dW_t,\quad t\in[0,T],
		\end{equation}
		where $\Psi_1(t,x)$ and $\Psi_2(t,x)$ are $\mathcal{F}_t^W$-predictable and satisfy standard integrability conditions.  
		Then, using the It\^o-Wentzell formula yields
		\begin{align*}
			du(t,X_t)
			=&\Psi_1(t,X_t)dt+\Psi_2(t,X_t)dW_t+\partial_x\Psi_2(t,X_t)d\langle W,X\rangle_t\\
			&+\partial_xu(t,X_t)dX_t+\frac{1}{2}\partial^2_{xx}u(t,X_t)d\langle X,X\rangle_t\\
			=&\Psi_1(t,X_t)dt+\Psi_2(t,X_t)dW_t+\partial_x\Psi_2(t,X_t)\langle f(t,X_t)\mathcal{D}_2\bs{p}_t,\widehat{\mathbb{W}}_t^{\infty}\rangle dt\\
			&+\partial_xu(t,X_t)\big(\langle\bs{q}^f_t,\widehat{\mathbb{W}}_t^{\infty}\rangle dt+\langle f(t,X_t)\mathcal{D}_2\bs{p}_t,\widehat{\mathbb{W}}_t^{\infty}\rangle dW_t+\langle g(t,X_t)\bs{\ell}_t,\widehat{\mathbb{W}}_t^{\infty}\rangle dB_t\big)\\
			&+\frac{1}{2}\partial^2_{xx}u(t,X_t)\langle f^2(t,X_t)(\mathcal{D}_2\bs{p}_t\shuffle\mathcal{D}_2\bs{p}_t)+g^2(t,X_t)(\bs{\ell}_t\shuffle\bs{\ell}_t),\widehat{\mathbb{W}}_t^{\infty}\rangle dt\\
			=&\Big[\Psi_1(t,X_t)+\partial_x\Psi_2(t,X_t)\langle f(t,X_t)\mathcal{D}_2\bs{p}_t,\widehat{\mathbb{W}}_t^{\infty}\rangle+\partial_xu(t,X_t)\langle\bs{q}_t^f,\widehat{\mathbb{W}}_t^{\infty}\rangle\\
			&+\frac{1}{2}\partial^2_{xx}u(t,X_t)\langle f^2(t,X_t)(\mathcal{D}_2\bs{p}_t\shuffle\mathcal{D}_2\bs{p}_t)+g^2(t,X_t)(\bs{\ell}_t\shuffle\bs{\ell}_t),\widehat{\mathbb{W}}_t^{\infty}\rangle\Big]dt\\
			&+\Big[\Psi_2(t,X_t)+\partial_xu(t,X_t)\langle f(t,X_t)\mathcal{D}_2\bs{p}_t,\widehat{\mathbb{W}}_t^{\infty}\rangle\Big]dW_t\\
			&+\partial_xu(t,X_t)\langle g(t,X_t)\bs{\ell}_t,\widehat{\mathbb{W}}_t^{\infty}\rangle dB_t, 
		\end{align*}
		where $\shuffle$ denotes the shuffle product\footnote{The shuffle product is commutative and associative, generating all possible interleavings of input tensors while preserving the relative order of indices from each original tensor. } for tensors (see, e.g.,~\cite{abi2025signature}, Section 2.1). 
		Since the value function $u(t,x)$ is defined as a conditional expectation, it constitutes a martingale with respect to $B$ by Doob's martingale property. 
		Therefore, by setting the coefficients of the $dW$ and $dt$ terms to $0$, we obtain
		\begin{align*}
			\Psi_1(t,X_t)=
			&-\partial_x\Psi_2(t,X_t)\langle f(t,X_t)\mathcal{D}_2\bs{p}_t,\widehat{\mathbb{W}}_t^{\infty}\rangle-\partial_xu(t,X_t)\langle\bs{q}^f_t,\widehat{\mathbb{W}}_t^{\infty}\rangle\\
			&-\frac{1}{2}\partial^2_{xx}u(t,X_t)\langle f^2(t,X_t)(\mathcal{D}_2\bs{p}_t\shuffle\mathcal{D}_2\bs{p}_t)+g^2(t,X_t)(\bs{\ell}_t\shuffle\bs{\ell}_t),\widehat{\mathbb{W}}_t^{\infty}\rangle,\\
			\Psi_2(t,X_t)=
			&-\partial_xu(t,X_t)\langle f(t,X_t)\mathcal{D}_2\bs{p}_t,\widehat{\mathbb{W}}_t^{\infty}\rangle.
		\end{align*}
		Substituting the coefficients $\Psi_1(t,X_t)$ and $\Psi_2(t,X_t)$ back into the random field~\eqref{eq-u-field} yields the following PDE with fixed and deterministic path $\omega$: 
		\begin{equation}
			\label{eq-RPDE-rep-t}
			\begin{cases}
				-du=&\big[\partial_xu\langle\bs{q}_t^f-f_0(t,x)(\mathcal{D}_2\bs{p}_t\shuffle\mathcal{D}_2\bs{p}_t),\widehat{\mathbb{W}}_t^{\infty}\rangle\\
				&+\frac{1}{2}\partial_{xx}^2u\langle g^2(t,x)(\bs{\ell}_t\shuffle\bs{\ell}_t)-f^2(t,x)(\mathcal{D}_2\bs{p}_t\shuffle\mathcal{D}_2\bs{p}_t),\widehat{\mathbb{W}}_t^{\infty}\rangle\big]dt\\
				&+\partial_xu\langle f(t,x)\mathcal{D}_2\bs{p}_t,\widehat{\mathbb{W}}_t^{\infty}\rangle dW_t,\\
				u(T,x)=&\Phi(x),
			\end{cases}
		\end{equation}
		where $f_0(t,x):=\partial_xf(t,x)f(t,x)$. 
		The PDE~\eqref{eq-RPDE-rep-t} is understood in a pathwise sense, meaning that $dW_t$ here is a fixed realization. 
		From the Feynman-Kac theorem and tower property, we know that any bounded solution $u(t,x)$ to this PDE has the unique representation
		\begin{equation*}
			\mathbb{E}[u(0,s_0)]=\mathbb{E}[\mathbb{E}[\Phi(X_T)\mid\widehat{\mathbb{W}}^{\infty}_T,X_0=s_0]]=\mathbb{E}[\Phi(X_T)].
		\end{equation*}
		Furthermore, recalling that $I_t:=\int_0^tv_sdW_s,\,t\in[0,T]$, the PDE \eqref{eq-RPDE-rep-t} can also be rewritten into a form that depends exclusively on the signature representation $\mathbf{v}_t=\langle\bs{\ell}_t,\widehat{\mathbb{W}}_t^{\infty}\rangle$ as follows
		\begin{equation*}
			\label{eq-RPDE-rep-t-degenerates}
			\begin{cases}
				-du=&\frac{1}{2}\partial^2_{xx}u(g^2(t,x)-f^2(t,x))\langle\bs{\ell}_t\shuffle \bs{\ell}_t,\widehat{\mathbb{W}}_t^{\infty}\rangle dt-\partial_xuf_0(t,x)\langle\bs{\ell}_t\shuffle\bs{\ell}_t,\widehat{\mathbb{W}}_t^{\infty}\rangle dt\\
				&+\partial_xuf(t,x)\langle\bs{\ell}_t,\widehat{\mathbb{W}}_t^{\infty}\rangle dW_t,\\
				u(T,x)=&\Phi(x).
			\end{cases}
		\end{equation*}
		
		\begin{remark}
			Following \cite{bank2025rough}, let $\mathbf{I}^g$ denote the Stratonovich rough path lift of the martingale $I$. The underlying asset price process \eqref{eq-sv} then admits an equivalent representation in a distributional sense:
			\begin{equation}
				\label{eq-RSDE-g}	
				dX_t=-\frac{1}{2}f_0(t,X_t)d[\mathbf{I}]_t+f(t,X_t)d\mathbf{I}^g_t+g(t,X_t)\mathbf{v}_tdB_t,\quad X_0=s_0,\; t\in[0,T],
			\end{equation}
			where the It\^o rough path $\mathbf{I}$ is uniquely identified with the pair $(\mathbf{I}^g,[\mathbf{I}])$ (see Lemma 2.8 in \cite{bank2025rough}). 
			Consequently, the European option price associated with the RSDE \eqref{eq-RSDE-g} satisfies the following RPDE:
			\begin{equation}
				\label{eq-RPDE-rep-t-degenerates-g}
				\begin{cases}
					-du=&\frac{1}{2}\big[\partial^2_{xx}ug^2(t,x)-\partial_xuf_0(t,x)\big]\mathbf{v}_t^2dt+\partial_xuf(t,x)d\mathbf{I}_t^g,\\
					u(T,x)=&\Phi(x).
				\end{cases}
			\end{equation}
			Because the RSDE \eqref{eq-RSDE-g} can be transformed into \eqref{eq-RSDE} via the relationship between the It\^o and Stratonovich integrals, it follows that the RPDE \eqref{eq-RPDE-rep-t-degenerates-g} is equivalent to the PDE \eqref{eq-RPDE-rep-t}. 		
		\end{remark}
		}

		\cpurple{As above, we replace the potentially non-Markovian pair $(v,I)$ with the rough pair $(\mathbf{v},\mathbf{I})$, and then represent it as an infinite linear combination of signatures. }
		However, in practice, it is impossible to choose the exact representation with an infinite sequence. Therefore, we truncate the infinite representation to finite linear combinations of signatures.
		Specifically, we introduce the projection $\pi_N(\cdot)$ for any $\bs{q}$ \cpurple{in an} extended tensor algebra  as follows
		\begin{equation*}
			\label{eq-truncation-projection}
			\pi_N(\bs{q})=(\pi_{(0)}(\bs{q}),\ldots,\pi_{(N)}(\bs{q}))=(\bs{q}^{(0)},\ldots,\bs{q}^{(N)}),
		\end{equation*}
		where $N\in\mathbb{N}$ denotes the truncation order.
		With the truncated sequences, the SDE~\eqref{eq-RSDE-rep-t} becomes
		\cpurple{
		\begin{equation}
			\label{eq-RSDE-rep-t-N}
			dX_t^N=\langle\pi_N(\bs{q}^f_t),\widehat{\mathbb{W}}_t^N\rangle dt+\langle f(t,X_t^N)\pi_N(\mathcal{D}_2\bs{p}_t),\widehat{\mathbb{W}}_t^N\rangle dW_t+\langle g(t,X_t^N)\pi_N(\bs{\ell}_t),\widehat{\mathbb{W}}_t^N\rangle dB_t,
		\end{equation}
		where
		\begin{equation*}
			\pi_N(\bs{q}^f_t):=f(t,X_t^N)\big[\pi_N(\dot{\bs{p}}_t)+\pi_N(\mathcal{D}_1\bs{p}_t)+\frac{1}{2}\pi_N(\mathcal{D}_{22}^2\bs{p}_t)\big],\quad t\in[0,T].
		\end{equation*}
		Here, the truncated coefficients of linear signature representations also lie in the admissible set $\mathcal{A}(\widehat{\mathbb{W}}^N)$. 
		Moreover, the corresponding PDE~\eqref{eq-RPDE-rep-t} becomes
		\begin{equation}
			\small
			\label{eq-RPDE-rep-t-N}
			\begin{cases}
				-du^N=&\big[\partial_xu^N[\langle\pi_N(\bs{q}_t^f),\widehat{\mathbb{W}}_t^N\rangle-\langle f_0(t,x)(\pi_N(\mathcal{D}_2\bs{p}_t)\shuffle\pi_N(\mathcal{D}_2\bs{p}_t)),\widehat{\mathbb{W}}_t^{2N}\rangle]\\
				&+\frac{1}{2}\partial_{xx}^2u^N\langle g^2(t,x)(\pi_N(\bs{\ell}_t)\shuffle\pi_N(\bs{\ell}_t))-f^2(t,x)(\pi_N(\mathcal{D}_2\bs{p}_t)\shuffle\pi_N(\mathcal{D}_2\bs{p}_t)),\widehat{\mathbb{W}}_t^{2N}\rangle\big]dt\\
				&+\partial_xu^N\langle f(t,x)\pi_N(\mathcal{D}_2\bs{p}_t),\widehat{\mathbb{W}}_t^N\rangle dW_t,\\
				u^N(T,x)=&\Phi(x).
			\end{cases}
		\end{equation}
		and $u^N(0,x)=\mathbb{E}\big[\Phi(X_T^N)\mid\widehat{\mathbb{W}}^{\infty}_T,X_0^N=x\big]$.
		The case of time-independent linear signature representation can be treated similarly by using truncated sequences. 
		}
		
	\subsection{Convergence analysis for linear signature representation}
	\label{sec-convergence}
		In this section, we study the gap between the finite and infinite linear signature representations, along with the solutions to the corresponding PDEs.
		To this end, we first establish the error bound for truncated linear signature combinations in Theorem~\ref{theorem-linear-rep-gap}, which serves as the foundation for our subsequent analysis.

		\begin{theorem}
			\label{theorem-linear-rep-gap}
			Let $(v_t)_{t\in[0,T]}$ be a $(\mathcal{F}_t^W)$-progressive and continuous stochastic process and $\langle\bs{\ell}_t,\widehat{\mathbb{W}}_t^{\infty}\rangle$ be the corresponding infinite linear signature representation.
			We define the finite linear signature representation of $(v_t)_{t\in[0,T]}$ by the first $N\in\mathbb{N}^+$ order components of its infinite linear signature representation.
			Then the gap between the finite and infinite linear signature representations, denoted by $G$, can be bounded by
			\begin{equation*}
				G_t^{\bs{\ell},N}:=\mathbb{E}\big|\langle\bs{\ell}_t,\widehat{\mathbb{W}}_t^{\infty}\rangle-\langle\pi_N(\bs{\ell}_t),\widehat{\mathbb{W}}_t^N\rangle\big|^2\leq\epsilon_t^{v,\widehat{\mathbb{W}},N},
			\end{equation*}
			where $\epsilon_t^{v,\widehat{\mathbb{W}},N}$, which depends on $t$, $v$, path signature of $W$ and the truncation order $N$, tends to zero as $N\to\infty$.		
		\end{theorem}
        \begin{proof}
            The result follows directly from Lemmas 3.2 and 3.3 of~\cite{bayraktar2024deep} by applying them with $k=1$ and taking $\Delta t\to0$.
        \end{proof}
		
		\cpurple{
		Theorem~\ref{theorem-linear-rep-gap} can be straightforwardly extended to the rough pair $(\mathbf{v},\mathbf{I})$, then we can give the error bounds for the solutions to~\eqref{eq-RSDE-rep-t} and~\eqref{eq-RSDE-rep-t-N}, as stated in Proposition~\ref{proposition-X-gap}.
		}

		\begin{proposition}
			\label{proposition-X-gap}
			Let $(\widehat{W}_t)_{t\in[0,T]}$ be a time-augmented one-dimensional Brownian motion.
			Assume that the deterministic functions $f$ and $g$ are locally Lipschitz continuous and satisfy the global linear growth conditions on $[0,T]$. 
			Under these assumptions, the gap between the infinite linear signature representation $X$ and its finite counterpart $X^N$ can be estimated by
			\begin{equation*}
				\label{eq-X-gap-l2}
				\mathbb{E}\big|X_t-X_t^N\big|^2\leq \epsilon_t^N,\quad t\in[0,T],
			\end{equation*}
			where $\epsilon_t^N\geq0$ is a sequence depending on $t$ and the truncation order $N$.
			In particular, $N\to\infty$ leads to $\epsilon_t^N\to0$.
		\end{proposition}
        \begin{proof}
            See Appendix \ref{appendA2}.
        \end{proof}
        	
		Having established the convergence properties of linear signature approximations, we now investigate how these representation gaps propagate to option pricing.
		Specifically, we quantify how the truncation errors in linear signature representations affect the gap between the solutions of the PDE~\eqref{eq-RPDE-rep-t} and its approximation~\eqref{eq-RPDE-rep-t-N}, as in Theorem~\ref{theorem-linear-option-price-gap}.

		\begin{theorem}
			\label{theorem-linear-option-price-gap}
			Assume that $\Phi:\mathbb{R}\to\mathbb{R}$ is Lipschitz continuous with constant $L_{\Phi}\geq0$ and solutions to SDEs~\eqref{eq-RSDE-rep-t} and~\eqref{eq-RSDE-rep-t-N} satisfy the bounds from Proposition \ref{proposition-X-gap}.
			The gap between solutions $u(0,s_0)$ and $u^N(0,s_0)$ to PDE~\eqref{eq-RPDE-rep-t} and~\eqref{eq-RPDE-rep-t-N}, respectively, can be estimated by
			\cpurple{
			\begin{equation*}
				\mathbb{E}\big|u(0,s_0)-u^N(0,s_0)\big|^2\leq L_{\Phi}^2\epsilon_T^N,
			\end{equation*}
			}
			where $\epsilon_T^N$ comes from Proposition~\ref{proposition-X-gap}.
			In particular, this gap decays to $0$ as $N\to\infty$.
		\end{theorem}
        \begin{proof}
            See Appendix \ref{appendA3}.
        \end{proof}

        We have established the convergence analysis for both time-dependent linear signature representations and their corresponding PDE solutions. The time-independent case follows analogously, thus requiring no further elaboration.

		\subsection{Nonlinear signature representation}
		{\color{purple}
		It should be noted that the preceding analysis is based on the restrictive assumption that the pair $(\mathbf{v},\mathbf{I})$ admits an infinite linear signature representation.
		To address the nonlinear path-dependent model, we introduce the deep nonlinear signature algorithm to approximate $(\mathbf{v},\mathbf{I})$. 
		Specifically, for $t\in[0,T]$, we define $\mathcal{N}^{\mathbf{v}}_t\big(\widehat{\mathbb{W}}^N_t;\,\phi_{\mathbf{v}}\big)$ and $\mathcal{N}^{\mathbf{I}}_t\big(\widehat{\mathbb{W}}^N_t;\,\phi_{\mathbf{I}}\big)$ as nonlinear neural networks parameterized by $\phi_{\mathbf{v}}$ and $\phi_{\mathbf{I}}$, respectively. 
		These neural networks take truncated signature inputs to approximate $\mathbf{v}_t$ and $\mathbf{I}_t$. 
		For simplicity, we suppress the arguments and denote them as $\mathcal{N}^{\mathbf{v},N}_t$ and $\mathcal{N}^{\mathbf{I},N}_t$ with no ambiguity. 
		As a consequence, the RSDE~\eqref{eq-RSDE} can be approximated by
		\begin{align}
			\label{eq-RSDE-rep-nonlinear}
			d\widehat{X}_t^N
			=&f(t,\widehat{X}_t^N)d\mathcal{N}_t^{\mathbf{I},N}+g(t,\widehat{X}_t^N)\mathcal{N}_t^{\mathbf{v},N}dB_t\nno\\
			=&f(t,\widehat{X}_t^N)\big(\langle\nabla\mathcal{N}_t^{\mathbf{I},N},d\widehat{\mathbb{W}}^N_t\rangle+\frac{1}{2}\langle\nabla^2\mathcal{N}_t^{\mathbf{I},N},d\widehat{\mathbb{W}}^N_t\otimes d\widehat{\mathbb{W}}^N_t\rangle\big)+g(t,\widehat{X}_t^N)\mathcal{N}_t^{\mathbf{v},N}dB_t\nno\\
			=&f(t,\widehat{X}_t^N)\big[\langle\nabla\mathcal{N}_t^{\mathbf{I},N},\pi_N(\widehat{\mathbb{W}}^N_t\otimes(e_1+\frac{1}{2}e_2\otimes e_2))\rangle dt+\langle\nabla\mathcal{N}_t^{\mathbf{I},N},\pi_N(\widehat{\mathbb{W}}^N_t\otimes e_2)\rangle dW_t\nno\\
			&+\frac{1}{2}\langle\nabla^2\mathcal{N}_t^{\mathbf{I},N},\pi_N(\widehat{\mathbb{W}}^N_t\otimes e_2)\otimes\pi_N(\widehat{\mathbb{W}}^N_t\otimes e_2)\rangle dt\big]+g(t,\widehat{X}_t^N)\mathcal{N}_t^{\mathbf{v},N}dB_t\nno\\
			=&\widehat{\bs{q}}^{f,N}_tdt+f(t,\widehat{X}_t^N)\langle\nabla\mathcal{N}_t^{\mathbf{I},N},\pi_N(\widehat{\mathbb{W}}^N_t\otimes e_2)\rangle dW_t+g(t,\widehat{X}_t^N)\mathcal{N}_t^{\mathbf{v},N}dB_t,
		\end{align}
		where
		\begin{align*}
			\label{eq-coefficient-q-t-nonlinear}
			\widehat{\bs{q}}^{f,N}_t:=&f(t,\widehat{X}_t^N)\big[\langle\nabla\mathcal{N}_t^{\mathbf{I},N},\pi_N(\widehat{\mathbb{W}}^N_t\otimes(e_1+\frac{1}{2}e_2\otimes e_2))\rangle\\
			&+\frac{1}{2}\langle\nabla^2\mathcal{N}_t^{\mathbf{I},N},\pi_N(\widehat{\mathbb{W}}^N_t\otimes e_2)\otimes\pi_N(\widehat{\mathbb{W}}^N_t\otimes e_2)\rangle\big].\nno
		\end{align*}
		In the linear case described earlier, the shuffle property was employed to streamline the computation of linear functionals. In the absence of the shuffle property, however, the PDE~\eqref{eq-RPDE-rep-t} and~\eqref{eq-RPDE-rep-t-N} can still be obtained through multiplication of the tensor inner products. Therefore, the nonlinear case can be analyzed analogously to obtain the truncated classical PDE via the Feynman-Kac theorem.
		}
		We then follow~\cite{bayraktar2024deep} to present the error bound for the truncated nonlinear signature representation in Theorem~\ref{theorem-nonlinear-rep}.
		\begin{theorem}
			\label{theorem-nonlinear-rep}
			\cpurple{
			Let $(v_t)_{t\in[0,T]}$ be a $(\mathcal{F}_t^W)$-progressive and continuous stochastic process, and  $\mathcal{N}^{\mathbf{v}}_t\big(\widehat{\mathbb{W}}_t^N;\phi_{\mathbf{v}}\big)$ denotes the neural network approximation of its It\^o rough path lift $(\mathbf{v})_{t\in[0,T]}$. 
			We define the approximation error as follows
			\begin{equation*}
				\label{}
				\mathcal{E}_t^{\mathcal{N},\mathbf{v}}:=\inf\limits_{\phi_{\mathbf{v}}}\mathbb{E}\big|\mathbf{v}_t-\mathcal{N}^{\mathbf{v}}_t\big(\widehat{\mathbb{W}}_t^N;\phi_{\mathbf{v}}\big)\big|^2.
			\end{equation*}
			We then have the following bound of the error:
			\begin{equation*}
		      \mathcal{E}^{\mathcal{N},\mathbf{v}}_t\leq\epsilon^{\mathcal{N},\mathbf{v}}_t+\epsilon^{\mathbf{v},\widehat{\mathbb{W}},N}_t,
			\end{equation*}
			where $\epsilon^{\mathcal{N},\mathbf{v}}_t$ denotes the error introduced from the neural network at time $t$, and $\epsilon^{\mathbf{v},\widehat{\mathbb{W}},N}_t$ denotes the error from the truncation of the signature at order $N$.
			}
		\end{theorem}
        \begin{proof}
            See Appendix \ref{appendA4}.
        \end{proof}
		The error bounds propagate to the solutions of the SDE~\eqref{eq-RSDE-rep-nonlinear} and ultimately to the corresponding PDE, following a pattern similar to the time-dependent linear case, as stated below.
		\begin{theorem}
			\label{theorem-nonlinear-option-price-gap}
			Let \cpurple{$\widehat{u}^{N}(0,s_0):=\mathbb{E}\big[\Phi\big(\widehat{X}_T^N\big)\mid\widehat{\mathbb{W}}^{\infty}_T,\widehat{X}_0^N=s_0\big]$}, where $\Phi$ is a Lipschitz continuous function with constant $L_{\Phi}\geq0$ and $\widehat{X}^N$ satisfies the SDE~\eqref{eq-RSDE-rep-nonlinear} with $N$-th order truncated signature inputs. Next, the gap between the solutions $u(0,s_0)$ and $\widehat{u}^N(0,s_0)$ can be estimated by
			\cpurple{
			\begin{equation*}
				\mathbb{E}\big|u(0,s_0)-\widehat{u}^N(0,s_0)\big|^2\leq L_{\Phi}^2\epsilon_T^{N,\mathcal{N}},
			\end{equation*}
			}
			where $\epsilon_T^{N,\mathcal{N}}\geq0$ is a sequence depending on $t$, the truncation order $N$ and the nonlinear neural networks $\mathcal{N}$.
		\end{theorem}
		\begin{proof}
			The proof is an application of Theorem~\ref{theorem-nonlinear-rep} and follows directly from the argument presented in Proposition~\ref{proposition-X-gap} and Theorem~\ref{theorem-linear-option-price-gap}, requiring no additional elaboration.
		\end{proof}
		
	\section{Numerical examples}
	\label{sec-examples}
		This section provides illustrative examples of European put option pricing using the linear and nonlinear signature representations of the pair $(\mathbf{v},\mathbf{I})$.
		\cpurple{
		All computations are performed using an NVIDIA RTX A5500 GPU. 
		The full implementation, data-generation scripts, and examples are available at~\href{https://github.com/changanluoxue/RSDE.git}{[GitHub link]}.
		}
		In these examples, we consider various stochastic volatility processes $v$, which are summarized in Table~\ref{tab-example-vol-processes}.
        For discretization, we partition the time interval $[0,T]$ into a uniform grid $\{t_j := j\Delta t,\, j=0,1,\ldots,J\}$, where $J\in\mathbb{N}^+$ and $\Delta t=T/J$.
		Moreover, we construct another uniform space-grid $\{x_l:=x_0+l\Delta x,\,l=0,1,\ldots,L\}$ on $[x_0,x_L]$ with $x_0,x_L\in\mathbb{R}$ and $x_0<x_L$, where $L\in\mathbb{N}^+$ and $\Delta x=\frac{x_L-x_0}{L}$.
		Using the established time-space grid, we give the finite-difference schemes for SDEs and PDEs with linear signature representations in the following subsection.

		\begin{table}[H]
			\caption{Examples of the stochastic volatility processes with their corresponding SDEs. Here, $v_t$ is the volatility with initial value $v_0>0$, $\alpha=\frac{1}{2}-H\in(0,1/2)$ and $\kappa,\theta,\sigma,\eta\in\mathbb{R}$.}
			\label{tab-example-vol-processes}
			\smallskip
			\centering
			\renewcommand\arraystretch{1.5}
			\begin{tabular}{l c}
				\toprule
				Volatility Process/Model&Stochastic Differential Equations\\
				\hline
				OU~\cite{uhlenbeck1930theory}&$dv_t=\kappa(\theta-v_t)dt+\eta dW_t$\\
				\hline
				mGBM~\cite{kluppelberg2004continuous}&$dv_t=\kappa(\theta-v_t)dt+(\eta+\sigma v_t)dW_t$\\
				\hline
				rHeston &$v_t=v_0+\int_0^tK(t-s)f(v_s)ds+\int_0^tK(t-s)g(v_s)dW_s,$\\
				\cite{el2019characteristic}&$K(t-s)=\frac{(t-s)^{-\alpha}}{\Gamma(1-\alpha)},\quad f(v_s)=\kappa(\theta-v_s),\quad g(v_s)=\sigma\sqrt{v_s}$\\
				\hline
				rBergomi~\cite{bayer2016pricing}&$v_t=v_0\exp{\left(\eta\int_0^t(t-s)^{-\alpha}dW_s\right)}$\\
				\bottomrule
			\end{tabular}
		\end{table}

		\subsection{Finite-difference schemes for the SDE and PDE with linear signature representations}
		In this section, we present the discretized SDE~\eqref{eq-RSDE-rep-t-N} and PDE~\eqref{eq-RPDE-rep-t-N} on the time-space grid points $(t_j,x_l)$, where $j=0,1,\ldots,J$ and $l=0,1,\ldots,L$.
		To simplify notation, we replace the discrete time and space stamps $t_j$ and $x_l$ with their index $j$ and $l$, respectively.
		
		Now, we can give the discretized SDE~\eqref{eq-RSDE-rep-t-N} as follows
		\cpurple{
		\begin{align}
			\label{eq-RSDE-rep-t-N-discrete}
			X_{j+1}^N
			=&X_{j}^N+\langle\pi_N(\bs{q}^f_j),\widehat{\mathbb{W}}_j^N\rangle\Delta t+\langle f(j,X_j^N)\pi_N(\mathcal{D}_2\bs{p}_j),\widehat{\mathbb{W}}_j^N\rangle (W_{j+1}-W_j)\\
			&+\langle g(j,X_j^N)\pi_N(\bs{\ell}_j),\widehat{\mathbb{W}}_j^N\rangle (B_{j+1}-B_j),\nno
		\end{align}
		where
		\begin{equation*}
			\pi_N(\bs{q}^f_j):=f(j,X_j^N)\left[\frac{\pi_N(\bs{p}_{j+1})-\pi_N(\bs{p}_j)}{\Delta t}+\pi_N(\mathcal{D}_1\bs{p}_j)+\frac{1}{2}\pi_N(\mathcal{D}_{22}^2\bs{p}_j)\right].
		\end{equation*}
		Here, we use time-dependent coefficients $\bs{\ell}_j$ and $\bs{p}_j$. 
		The time-independent case can be derived similarly and straightforwardly.
		}
		
		Next, we discretize the PDE~\eqref{eq-RPDE-rep-t-N} by the first-order finite difference scheme with additional Dirichlet boundary conditions. To achieve this, we first define the boundary functions $\psi_0,\,\psi_L:\,[0,T]\to\mathbb{R}$ as follows
		\begin{eqnarray*}
			&u(t,x_0)=\psi_0(t),\; u(t,x_L)=\psi_L(t),\; u(T,x)=\Phi(x),\quad (t,x)\in[0,T]\times[x_0,x_L].
		\end{eqnarray*}
		For simplicity, the truncation order argument for $u$ is omitted here.
		Then we employ the Crank-Nicolson method for discretization, as the explicit scheme was observed to be unstable unless very small time steps were used.
		Thus, we approximate the spatial derivatives in PDE~\eqref{eq-RPDE-rep-t-N} by central finite difference quotients as follows
		\begin{equation*}
			\label{eq-derivatives}
			\delta_xu_j^l:=\frac{u_j^{l+1}-u_j^{l-1}}{2\Delta x},\qquad
			\delta^2_{xx}u_j^l:=\frac{u_j^{l+1}+u_j^{l-1}-2u_j^l}{(\Delta x)^2}
		\end{equation*}
		for $j=0,1,\ldots,J$ and $l=1,2,\dots,L-1$.
		From all things above, the Crank-Nicolson discretization reads
		\cpurple{
		\begin{flalign}
			\label{eq-RPDE-rep-t-N-discrete}
			-(u_{j+1}^l-u_j^l) =\frac{\Delta t}{2}(\mathcal{L}_j^l+\mathcal{L}_{j+1}^l)+\delta_xu_{j+1}^l\langle f(t_{j+1},x_l)\pi_N(\mathcal{D}_2\bs{p}_{j+1}),\widehat{\mathbb{W}}_{j+1}^N\rangle(W_{j+1}-W_j)
		\end{flalign}
		for $0\leq j\leq J-1$, $1\leq l\leq L-1$, with boundary conditions
		\begin{equation*}
			u_j^0=\psi_0(t_j),\;u_j^L=\psi_L(t_j),\;u_J^l=\Phi(x_l).
		\end{equation*}
		The operator $\mathcal{L}$ reads
		\begin{align*}
			\mathcal{L}_j^l
			=&\delta_xu_j^l[\langle\pi_N(\bs{q}_j^f),\widehat{\mathbb{W}}_j^N\rangle-\langle f_0(t_j,x_l)(\pi_N(\mathcal{D}_2\bs{p}_j)\shuffle\pi_N(\mathcal{D}_2\bs{p}_j)),\widehat{\mathbb{W}}_j^{2N}\rangle]\\
			&+\frac{1}{2}\delta^2_{xx}u_j^l\langle g^2(t_j,x_l)(\pi_N(\bs{\ell}_j)\shuffle\pi_N(\bs{\ell}_j))-f^2(t_j,x_l)(\pi_N(\mathcal{D}_2\bs{p}_j)\shuffle\pi_N(\mathcal{D}_2\bs{p}_j)),\widehat{\mathbb{W}}_j^{2N}\rangle.
		\end{align*} 
		}
		
		Following the above discretization framework, we now present the first example with the Markovian volatility processes: OU and mGBM, from Table~\ref{tab-example-vol-processes}.

		\begin{example}[Markovian volatility processes with linear signature representations]
		\label{example-Markovian}
			\cpurple{
			In this example, we first analytically reconstruct the rough pair $(\mathbf{v},\mathbf{I})$ for the Markovian volatility processes $v_t$ and its integrated counterpart $I_t:=\int_0^tv_sdW_s$ using linear signature representations with time-independent coefficients $\bs{\ell}$ and $\bs{p}$, respectively.}
			We then incorporate them into the above discretization framework to compute European put option prices. 
			\cpurple{
			Specifically, we can estimate prices either by simulating paths via the discretized SDE~\eqref{eq-RSDE-rep-t-N-discrete} to evaluate the conditional expectation $\mathbb{E}\big[\Phi(X_T)\mid\widehat{\mathbb{W}}^N_T\big]$ (SDE approach), or by computing the solution to PDE~\eqref{eq-RPDE-rep-t-N-discrete} directly (PDE approach). These two approaches are equivalent due to the Feynman-Kac theorem. Finally, we use the tower property to get the unconditional expectation $\mathbb{E}\big[\Phi(X_T)\big]$.
			}
			
			For implementation details, we consider the case where the asset-price dynamics~\eqref{eq-sv} follow the SABR dynamics, i.e., $f(t,x)=\rho x^{\beta}$ and $g(t,x)=\sqrt{1-\rho^2}x^{\beta}$ for some $\beta\in\left(\frac{1}{2},1\right)$. And the volatility process $v$ follows either the OU or mGBM process as specified in Table~\ref{tab-example-vol-processes}.
			We set the time and space step size as $\Delta t=1/251$ and $\Delta x=1/400$, respectively.
			The parameters for the SABR stochastic volatility model and European put options are specified as follows
			\cpurple{
			\begin{eqnarray*}
                \text{OU}&\colon&\kappa=1,\;\theta=0.25,\;\eta=1.2,\;v_0=0.1,\;\rho=-0.4, \;\beta=0.6,\;K=110,\;T=1;\\
                \text{mGBM}&\colon&\kappa=1,\;\theta=0.25,\;\sigma=0.5,\;\eta=0,\;v_0=0.1,\;\rho=-0.4, \;\beta=0.6,\;K=110,\;T=1.
			\end{eqnarray*}}
			where $K$ is the strike price.
		\end{example}
		
		We begin by computing the time-independent coefficients $\bs{\ell}$ and $\bs{p}$ of the linear signature representations $\widetilde{\mathbf{v}}$ and $\widetilde{\mathbf{I}}$.
		We first consider the case where $v$ follows the OU process, then the linear signature representation $\widetilde{\mathbf{v}}$ has coefficients of the following algebraic form
		\begin{equation*}
			\widetilde{\mathbf{v}}_t=\langle\bs{\ell}^{OU},\widehat{\mathbb{W}}_t^{\infty}\rangle,\qquad
			\bs{\ell}^{OU}=(v_0\cblue{\bs{\Phi}}+\kappa\theta\cblue{\bs{1}}+\eta\cblue{\bs{2}})e^{\shuffle-\kappa\cblue{\bs{1}}}.
		\end{equation*}
		To be more explicit, the linear form of $\bs{\ell}^{OU}$ reads
		\begin{equation*}
			\bs{\ell}^{OU}
			=\left(v_0,
			\begin{pmatrix}
				-\kappa(v_0-\theta) \\ \eta
			\end{pmatrix},
			\begin{pmatrix}
				\kappa^2(v_0-\theta)&0\\
				-\kappa\eta&0
			\end{pmatrix},
			\begin{pmatrix}
				-\kappa^3(v_0-\theta)&0&\\
				0&0&\\
				&\kappa^2\eta&0\\
				&0&0
			\end{pmatrix},\ldots
			\right),
		\end{equation*}
		where the blue letters denote the tensor product of the canonical basis, i.e., $\cblue{\bs{1}}\cdots\cblue{\bs{d}}:=e_1\otimes\cdots\otimes e_d$ for $d=1,2,\ldots$ (see, e.g.,~\cite{abi2025signature}, Section 2.1).
		With Proposition~\ref{proposition-time-independent}, we can further give the coefficient $\bs{p}^{OU}$ of the linear representation $\widetilde{I}$ as follows
		\begin{flalign*}
			\bs{p}^{OU}
			=&v_0\cdot\cblue{\bs{2}}-\kappa(v_0-\theta)\cdot\cblue{\bs{12}}+\eta\cdot\cblue{\bs{22}}+\kappa^2(v_0-\theta)\cdot\cblue{\bs{112}}-\kappa\eta\cdot\cblue{\bs{212}}\\
			& - \kappa^3(v_0-\theta)\cdot\cblue{\bs{1112}} + \kappa^2\eta\cdot\cblue{\bs{2112}} + \kappa^4(v_0-\theta)\cdot\cblue{\bs{11112}} - \kappa^3\eta\cdot\cblue{\bs{21112}} + \cdots.
		\end{flalign*}
		
		Similarly, the mGBM process also admits an infinite linear signature representation with coefficients of the following algebraic form
		\begin{equation*}
			\widetilde{\mathbf{v}}_t=\langle\bs{\ell}^{mGBM}, \widehat{\mathbb{W}}_t^{\infty}\rangle,\quad \bs{\ell}^{mGBM}=\left(v_0\cblue{\bs{\Phi}}+\left(\kappa\theta-\frac{\sigma\eta}{2}\right)\cblue{\bs{1}}+\eta\cblue{\bs{2}}\right)e^{\shuffle\left(-\left(\kappa+\frac{\sigma^2}{2}\right)\cblue{\bs{1}}+\sigma\cblue{\bs{2}}\right)}.
		\end{equation*}
		The linear form of $\bs{\ell}^{mGBM}$ reads
		\begin{flalign*}
			\bs{\ell}^{mGBM}&=\left( v_0,
			\begin{pmatrix}
				v_0\lambda+\gamma \\ v_0\sigma+\eta
			\end{pmatrix} ,
			\begin{pmatrix}
				\lambda(v_0\lambda+\gamma)&\sigma(v_0\lambda+\gamma)\\
				\lambda(v_0\sigma+\eta)&\sigma(v_0\sigma+\eta)
			\end{pmatrix} ,\right.\\
			&\left.\begin{pmatrix}
				\lambda^2(v_0\lambda+\gamma)&\lambda\sigma(v_0\lambda+\gamma)&\\
				\lambda\sigma(v_0\lambda+\gamma)&\sigma^2(v_0\lambda+\gamma)&\\
				&\lambda^2(v_0\sigma+\eta)&\lambda\sigma(v_0\sigma+\eta)\\
				&\lambda\sigma(v_0\sigma+\eta)&\sigma^2(v_0\sigma+\eta)
			\end{pmatrix} ,\ldots
			\right),
		\end{flalign*}
		where $\lambda=-\left(\kappa+\frac{\sigma^2}{2}\right)$ and $\gamma=\kappa\theta-\frac{\sigma\eta}{2}$. 
		We can further derive the coefficient $\bs{p}^{mGBM}$ as follows
		\begin{flalign*}
			&\bs{p}^{mGBM}
			=v_0\cdot\cblue{\bs{2}}+(v_0\lambda+\gamma)\cdot\cblue{\bs{12}}+(v_0\sigma+\eta)\cdot\cblue{\bs{22}}\\
			&+\lambda(v_0\lambda+\gamma)\cdot\cblue{\bs{112}}+\sigma(v_0\lambda+\gamma)\cdot\cblue{\bs{122}}+\lambda(v_0\sigma+\eta)\cdot\cblue{\bs{212}}+\sigma(v_0\sigma+\eta)\cdot\cblue{\bs{222}}\\
			&+\lambda^2(v_0\lambda+\gamma)\cdot\cblue{\bs{1112}}+\lambda\sigma(v_0\lambda+\gamma)\cdot\cblue{\bs{1122}}\\
			&+\lambda\sigma(v_0\lambda+\gamma)\cdot\cblue{\bs{1212}}+\sigma^2(v_0\lambda+\gamma)\cdot\cblue{\bs{1222}}+\cdots.
		\end{flalign*}
		
		Now, having these explicit coefficients in hand, we can analytically reconstruct $\mathbf{v}$ and $\mathbf{I}$ with linear signature representations $\widetilde{\mathbf{v}}$ and $\widetilde{\mathbf{I}}$, respectively. We employ a Monte Carlo method with the path number $M=10,000$ and the number of time intervals $J=251$ as our benchmark and evaluate the path-wise mean absolute error (MAE)  $\epsilon$ and the overall MAE $\mathcal{E}$ as follows
		\begin{eqnarray*}
		\label{eq-MAE}
			\epsilon(A_{m,\cdot},\widetilde{A}_{m,\cdot})&:=&\frac{1}{J+1}\sum\limits_{j=0}^J\big|A_{m,j}-\widetilde{A}_{m,j}\big|,\,m=1,\ldots,M,\\
			\mathcal{E}(A,\widetilde{A})&:=&\frac{1}{M}\sum\limits_{m=1}^M\epsilon(A_{m,\cdot},\widetilde{A}_{m,\cdot}),
		\end{eqnarray*}
		where $A$ and $\widetilde{A}$ represent arbitrary matrices containing $M$ paths, each of length $J+1$.
		In Figure~\ref{fig-example1-v-I-box}, we plot the distribution of path-wise MAEs $\epsilon$ and present the overall MAEs $\mathcal{E}$ between the Monte Carlo sample paths $(\mathbf{v},\mathbf{I})$ and their linear signature representations $(\widetilde{\mathbf{v}},\widetilde{\mathbf{I}})$ at different truncation levels for both OU and mGBM processes. The corresponding values are detailed in Table~\ref{tab-example1-v-I-MAE}.
		\cpurple{
		The results show that the accuracy and robustness of linear signature representations $\widetilde{\mathbf{v}}$ and $\widetilde{\mathbf{I}}$ (for both OU and mGBM) increases significantly as the truncation level increases. 
		Moreover, those for OU exhibit greater sensitivity to truncation levels than those for mGBM. 
		These indicate that the accuracy and robustness of linear signature representations is sensitive to both the signature truncation level and the type of volatility process.
		}
		
		\begin{figure}[H]
			\centering
			\subfigure[OU]{\includegraphics[width=0.75\textwidth]{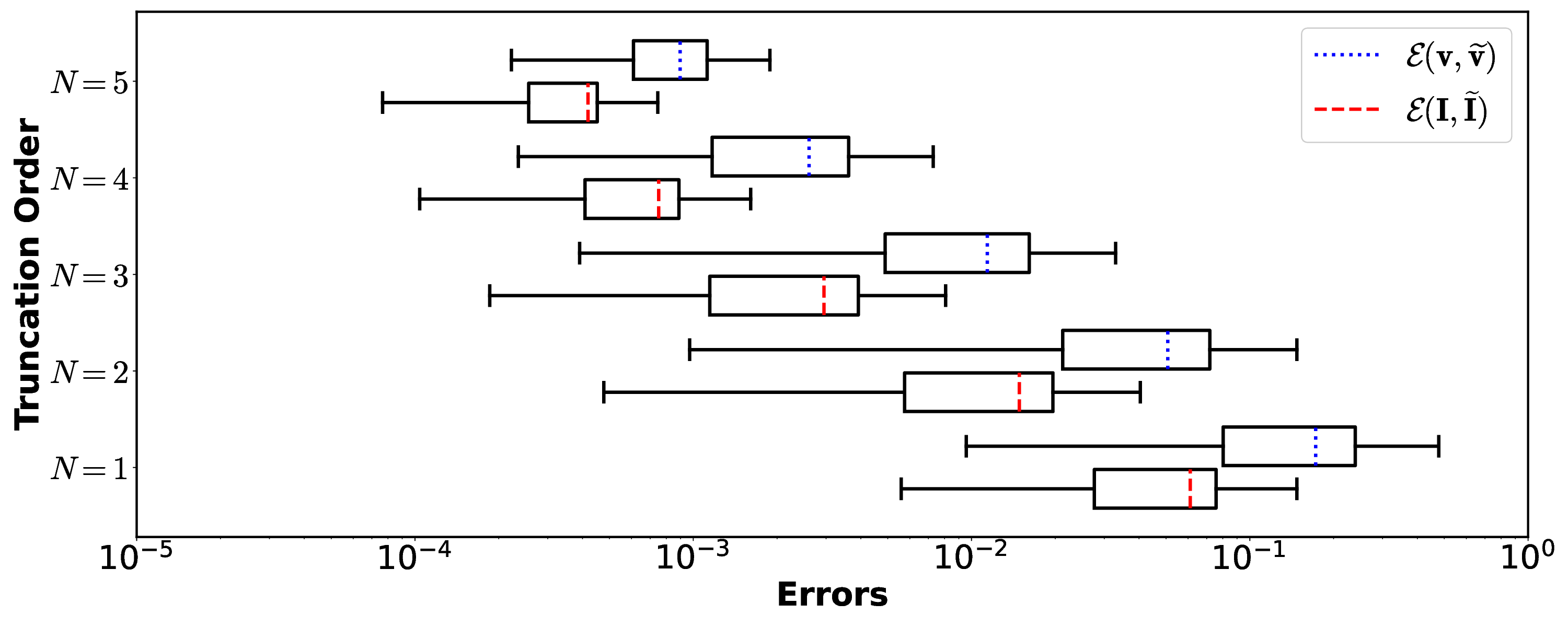}}
			\subfigure[mGBM]{\includegraphics[width=0.75\textwidth]{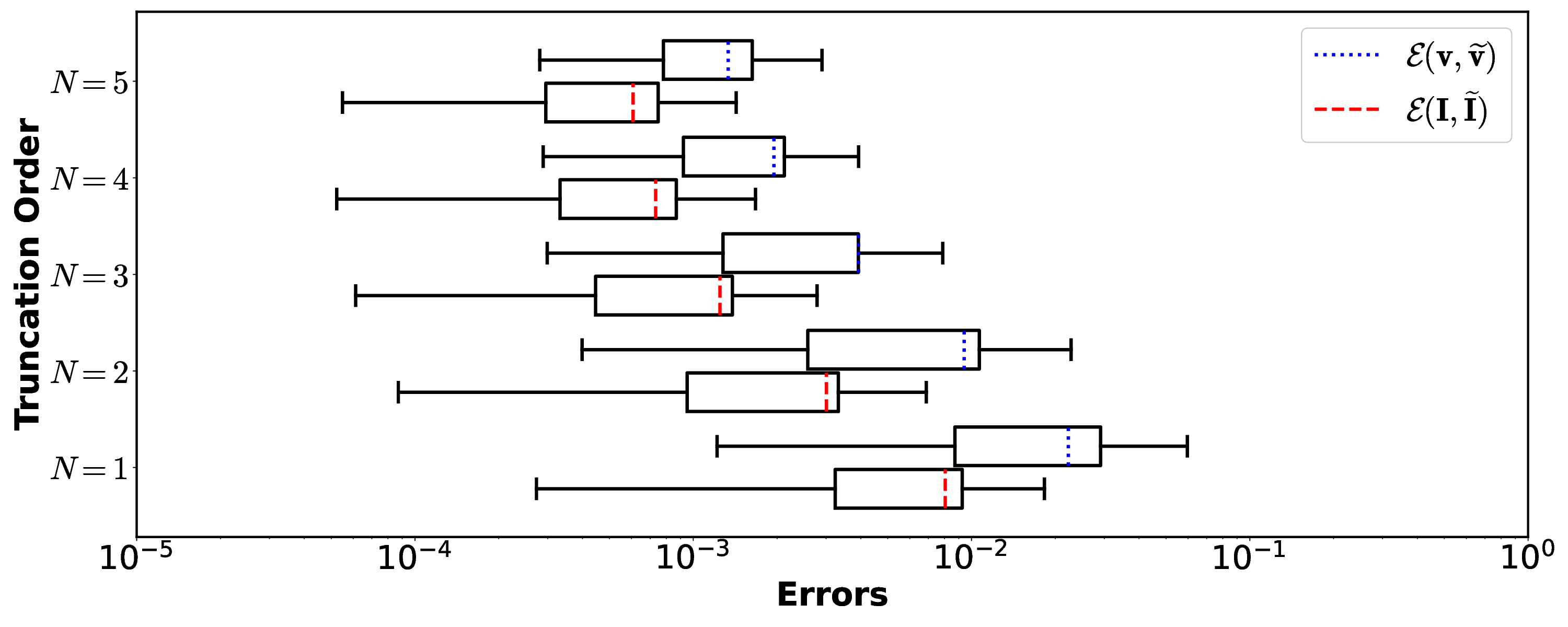}}
			\caption{Accuracy of linear signature representations $(\widetilde{\mathbf{v}},\widetilde{\mathbf{I}})$ for truncation levels $N=1$ to $5$ under the OU process and mGBM: Each box spans from the first quartile ($25\%$) to the third quartile ($75\%$) of path-wise MAEs $\epsilon$. The whisker boundaries denote the maximum and minimum $\epsilon$. The overall MAEs $\mathcal{E}$ are marked out by the dashed (or dotted) lines. The horizontal axis uses a logarithmic scale, and the results are based on the Monte Carlo method with path number $M=10,000$ and the number of time intervals $J=251$.}
			\label{fig-example1-v-I-box}
		\end{figure}
		
		\begin{table}[H]
			\caption{Overall MAEs $\mathcal{E}$ between $(\mathbf{v},\mathbf{I})$ and $(\widetilde{\mathbf{v}},\widetilde{\mathbf{I}})$ corresponding to Figure~\ref{fig-example1-v-I-box} for different truncation levels and volatility processes. Standard deviations of path-wise MAEs $\epsilon$ appear in parentheses.}
			\label{tab-example1-v-I-MAE}
			\smallskip
			\centering
			\renewcommand\arraystretch{1.3}
			\setlength{\tabcolsep}{4pt}
			\begin{tabular}{l c c c c}
				\toprule
				\multirow{2}{*}{Truncation Level} & \multicolumn{2}{c}{Errors for OU process} & \multicolumn{2}{c}{Errors for mGBM}\\
				\cmidrule{2-5}
				& $\mathcal{E}(\mathbf{v},\widetilde{\mathbf{v}})$ & $\mathcal{E}(\mathbf{I},\widetilde{\mathbf{I}})$ & $\mathcal{E}(\mathbf{v},\widetilde{\mathbf{v}})$ & $\mathcal{E}(\mathbf{I},\widetilde{\mathbf{I}})$ \\
				\midrule
				\multirow{1}{*}{$N=1$} & 1.72e-1 (1.18e-1) & 6.11e-2 (5.24e-2) & 2.23e-2 (1.98e-2) & 8.05e-3 (9.11e-3)\\
				\multirow{1}{*}{$N=2$} & 5.07e-2 (3.65e-2) & 1.49e-2 (1.33e-2) & 9.41e-3 (1.19e-2) & 3.01e-3 (4.38e-3)\\
				\multirow{1}{*}{$N=3$} & 1.14e-2 (8.11e-3) & 2.95e-3 (2.58e-3) & 3.92e-3 (5.49e-3) & 1.25e-3 (1.84e-3) \\
				\multirow{1}{*}{$N=4$} & 2.61e-3 (1.80e-3) & 7.52e-4 (5.82e-4) & 1.95e-3 (2.22e-3) & 7.33e-4 (7.89e-4)\\
				\multirow{1}{*}{$N=5$} & 8.97e-4 (3.80e-4) & 4.19e-4 (3.17e-4) & 1.34e-3 (9.37e-4) & 6.08e-4 (5.14e-4)\\
				\bottomrule
			\end{tabular}
		\end{table}
		
        As outlined previously, we now approximate the European put option prices using $(\widetilde{\mathbf{v}},\widetilde{\mathbf{I}})$ by substituting them into both the discrete SDE~\eqref{eq-RSDE-rep-t-N-discrete} and the PDE~\eqref{eq-RPDE-rep-t-N-discrete}.\cpurple{
		 That is, we employ SDE~\eqref{eq-RSDE-rep-t-N-discrete} for Monte Carlo simulation to compute the expectation $\mathbb{E}\big[\Phi(X_T)\big]$, while utilizing the expected solution to PDE~\eqref{eq-RPDE-rep-t-N-discrete} as the option price through the Feynman-Kac theorem.
		As before, we use Monte Carlo-simulated option prices $\mathbb{E}\big[\Phi(S_T)\big]$ as the benchmark, measuring approximation accuracy through the absolute error
		\begin{equation*}
			\big|\mathbb{E}\big[\Phi(S_T)\big]-\mathbb{E}\big[\Phi(X_T)\big]\big|.
		\end{equation*}
		}
		
		Figure~\ref{fig-example1-option-price} shows the approximation errors in option pricing for out-of-the-money (OTM), at-the-money (ATM), and in-the-money (ITM) cases under both the OU process and mGBM across different truncation levels. Detailed results are presented in Table~\ref{tab-example1-option-price}.
		\begin{figure}[H]
			\centering
			\subfigure[OU]{
				\includegraphics[width=0.48\textwidth]{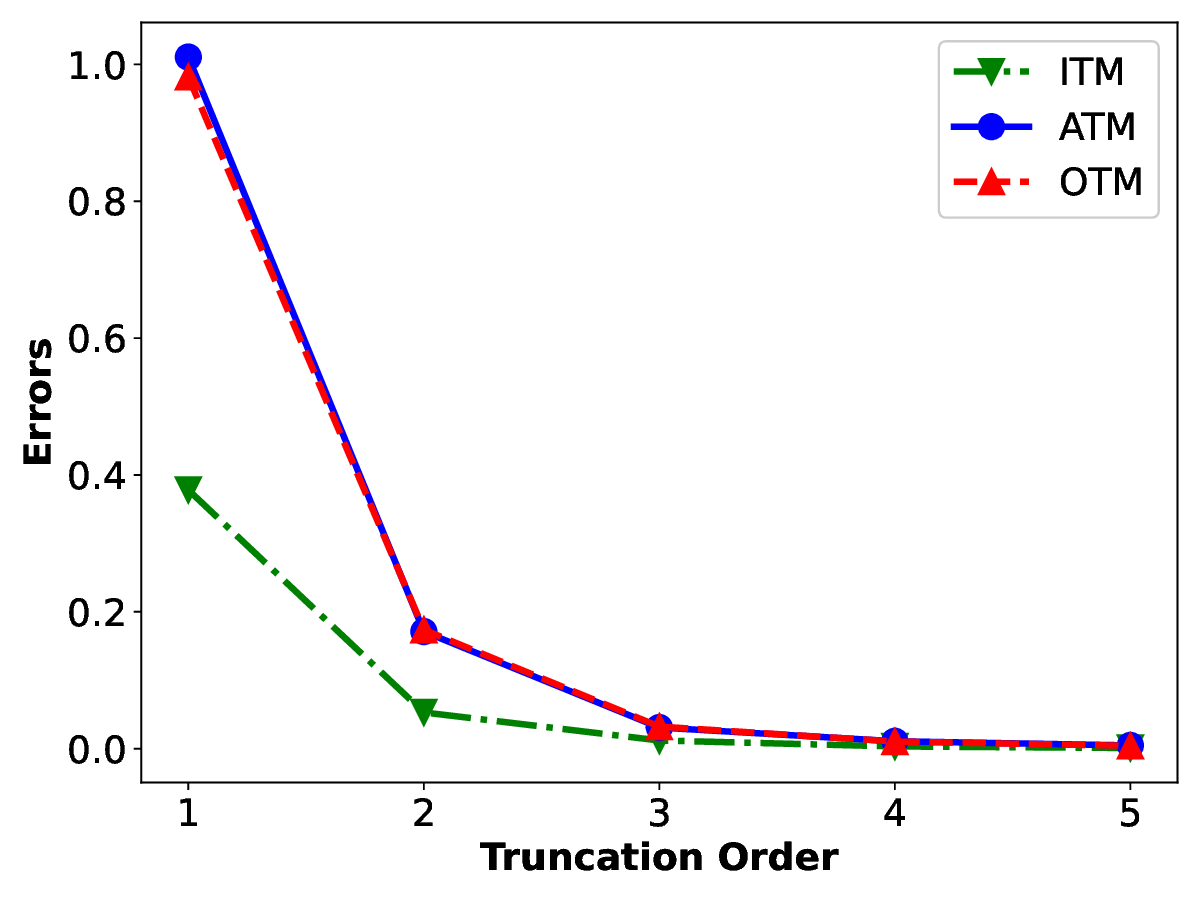}
				\includegraphics[width=0.48\textwidth]{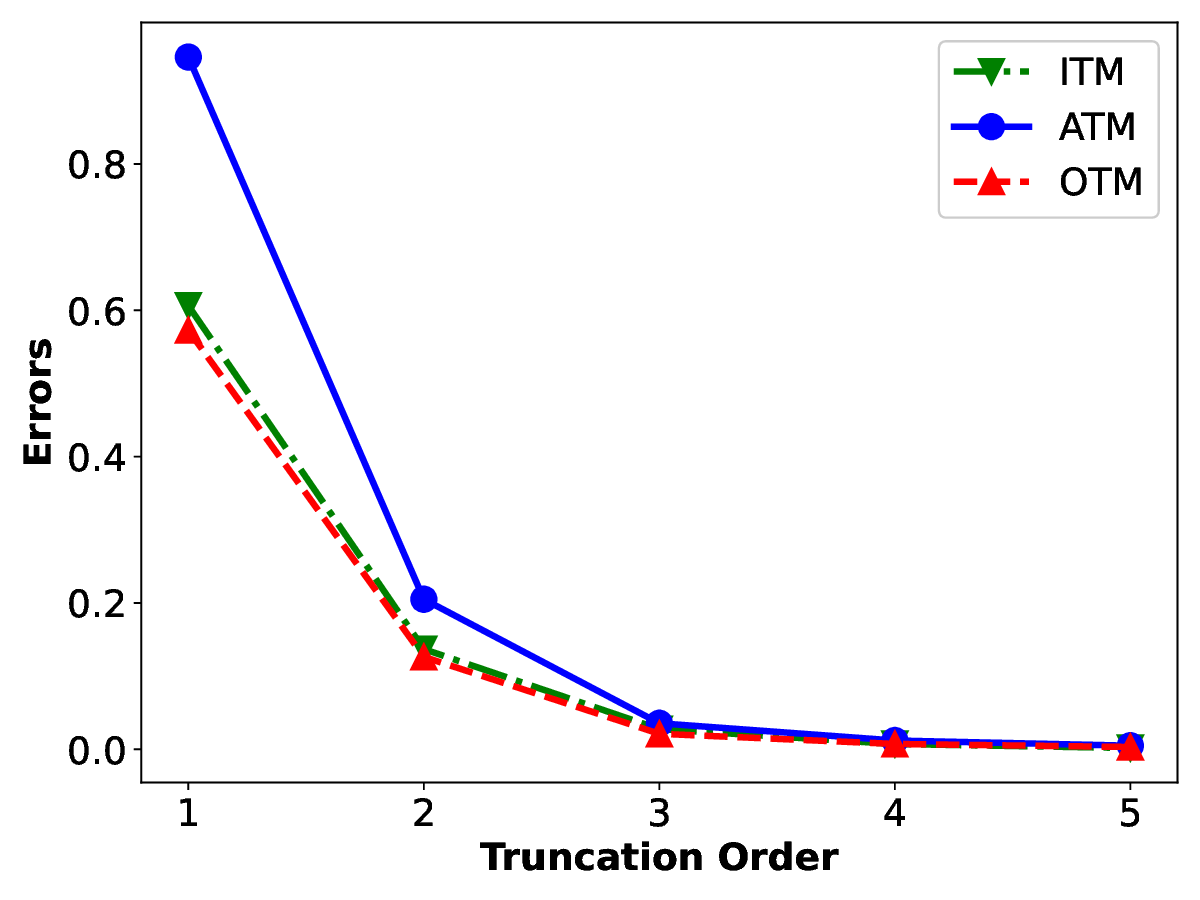}
			}
            \vspace{-0.5cm}
			\subfigure[mGBM]{
				\includegraphics[width=0.48\textwidth]{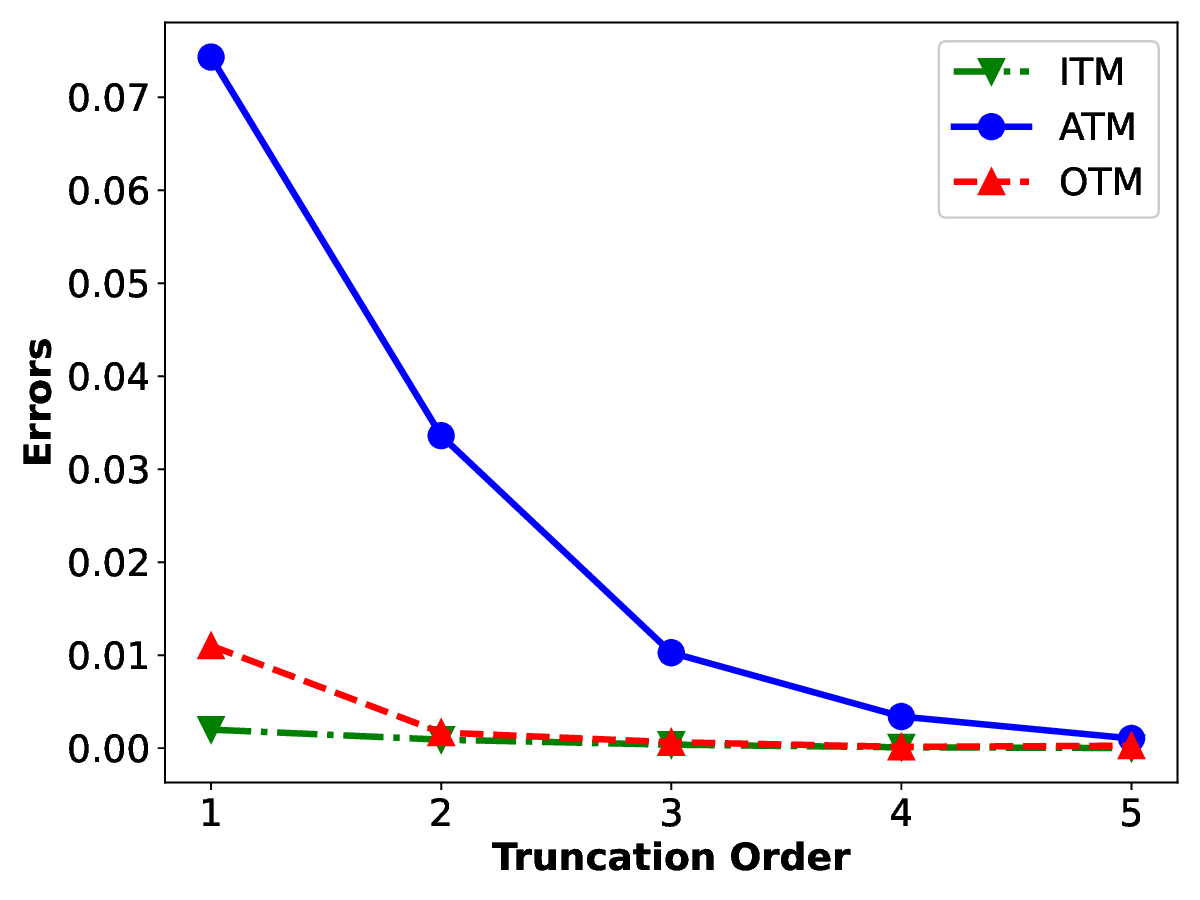}
				\includegraphics[width=0.48\textwidth]{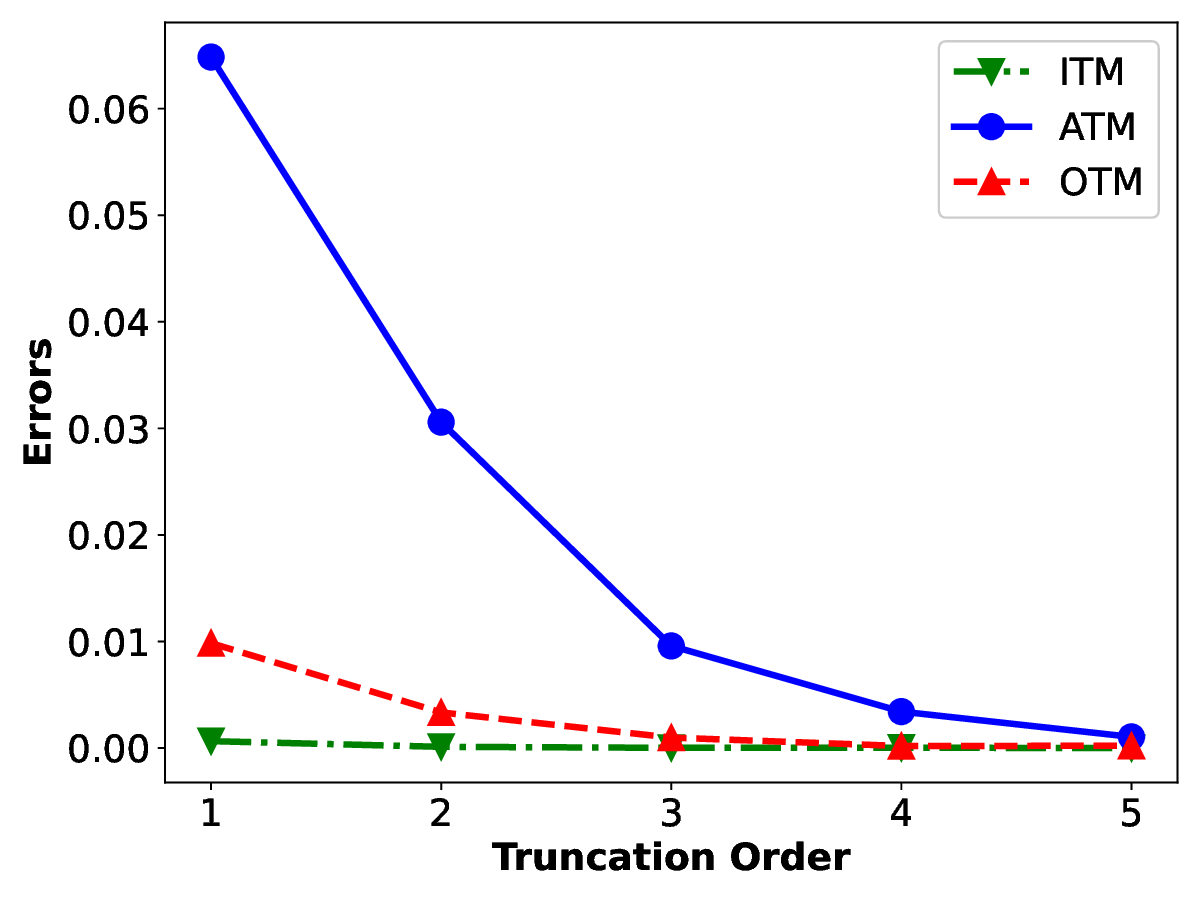}
			}
			\caption{Approximation accuracy for option prices at different moneyness levels. The left column are the errors for SDE~\eqref{eq-RSDE-rep-t-N-discrete} and the right column are the errors for PDE~\eqref{eq-RPDE-rep-t-N-discrete}. The analysis considers three cases: ATM ($x_0=110=K$), ITM ($x_0=95<K$), and OTM ($x_0=115>K$).}
			\label{fig-example1-option-price}
		\end{figure}
		 
		\begin{table}[H]
			\caption{The approximation errors corresponding to Figure~\ref{fig-example1-option-price} for different truncation levels and volatility processes.}
			\label{tab-example1-option-price}
			\smallskip
			\centering
			\renewcommand\arraystretch{1.3}
			\setlength{\tabcolsep}{2.6pt}
			\begin{tabular}{l c c c c c c c}
				\toprule
				\multirow{2}{*}{Model} & \multirow{2}{*}{Truncation Level} & \multicolumn{3}{c}{Errors obtained from SDE~\eqref{eq-RSDE-rep-t-N-discrete}} & \multicolumn{3}{c}{Errors obtained from PDE~\eqref{eq-RPDE-rep-t-N-discrete}}\\
				\cmidrule{3-8}
				&& ITM & ATM & OTM & ITM & ATM & OTM \\
				\midrule
				\multirow{5}{*}{OU}
				& $N=1$ & 3.78e-1 & 1.01e+0 & 9.82e-1 & 6.06e-1 & 9.46e-1 & 5.73e-1\\
				& $N=2$ & 5.30e-2 & 1.71e-1 & 1.74e-1 & 1.37e-1 & 2.05e-1 & 1.26e-1\\
				& $N=3$ & 1.20e-2 & 3.09e-2 & 3.23e-2 & 2.76e-2 & 3.56e-2 & 2.12e-2\\
				& $N=4$ & 3.26e-3 & 1.09e-2 & 1.05e-2 & 7.25e-3 & 1.20e-2 & 7.36e-3\\
				& $N=5$ & 9.44e-4 & 4.79e-3 & 4.16e-3 & 1.76e-3 & 4.95e-3 & 3.17e-3\\
				\hline
				\multirow{5}{*}{mGBM}
				& $N=1$ & 2.00e-3 & 7.43e-2 & 1.10e-2 & 6.49e-4 & 6.48e-2 & 9.84e-3\\
				& $N=2$ & 9.14e-4 & 3.36e-2 & 1.68e-3 & 1.17e-4 & 3.06e-2 & 3.35e-3\\
				& $N=3$ & 3.73e-4 & 1.03e-2 & 6.39e-4 & 2.20e-5 & 9.59e-3 & 9.83e-4\\
				& $N=4$ & 9.20e-5 & 3.40e-3 & 1.43e-4 & 2.90e-5 & 3.42e-3 & 1.94e-4\\
				& $N=5$ & 1.40e-5 & 1.04e-3 & 2.66e-4 & 4.00e-6 & 1.05e-3 & 2.18e-4\\
				\bottomrule
			\end{tabular}
		\end{table}	
        We have four key observations: (1) Both numerical schemes show monotonically decreasing errors with higher truncation levels; \cpurple{(2) The accuracy for mGBM consistently outperforms that for OU across all scenarios; (3) The approximation accuracy is sensitive to moneyness, generally producing the largest errors for ATM options; (4) The SDE and PDE methods perform comparably overall. }
		
		\subsection{Deep learning schemes for linear and nonlinear signature representations}
		As discussed in Section~\ref{sec-problem}, we employ neural networks $(\mathcal{N}^{\mathbf{v},N},\mathcal{N}^{\mathbf{I},N})$ as the nonlinear signature approximation of $(\mathbf{v},\mathbf{I})$.
		Specifically, for each time step $j$, we define the map $\mathcal{N}^{\mathbf{v}}:\big(j,\widehat{\mathbb{W}}_j^N;\phi_{\mathbf{v}}\big)\mapsto \widehat{\mathbf{v}}_j$ to directly learn $\mathbf{v}_j$, and its training problem reads
		\begin{equation*}
			\min\limits_{\phi_{\mathbf{v}}}\frac{1}{J}\sum\limits_{j=0}^J\big(\mathcal{N}^{\mathbf{v},N}_j-\mathbf{v}_j\big)^2.
		\end{equation*}
		The architecture of $\mathcal{N}^{\mathbf{v},N}$ comprises a fully connected feedforward neural network that has 5 hidden layers with 32 neurons each and uses ReLU activation functions. 
		\cpurple{
		Similarly, we define the neural network $\mathcal{N}^{\mathbf{I}}:\big(j,\widehat{\mathbb{W}}_j^N;\,\phi_{\mathbf{I}}\big)\mapsto \widehat{\mathbf{I}}_j$, where $\widehat{\mathbf{I}}_j$ is the neural network output to learn $\mathbf{I}_j$. 
		}
		We further give the discrete truncated SDE and its associated PDE as follows
		\cpurple{
		\begin{align}
			\label{eq-nn-SDE}
			\widehat{X}_{j+1}^N
			=&\widehat{X}_j^N
			+\widehat{\bs{q}}_j^{f,N}\Delta t
			+f(j,\widehat{X}^N_j)\langle\nabla\mathcal{N}_j^{\mathbf{I},N},\pi_N(\widehat{\mathbb{W}}_j^N\otimes e_2)\rangle(W_{j+1}-W_j)\\
			&+g(j,\widehat{X}^N_j)\mathcal{N}_j^{\mathbf{v},N}(B_{j+1}-B_j),\nno
		\end{align}
		\begin{flalign}
			\label{eq-nn-PDE}
			-(\widehat{u}_{j+1}^l-\widehat{u}_j^l) =&\frac{\Delta t}{2}(\widehat{\mathcal{L}}_j^l+\widehat{\mathcal{L}}_{j+1}^l)+\delta_x\widehat{u}_{j+1}^lf(t_{j+1},x_l)\langle\nabla\mathcal{N}_{j+1}^{\mathbf{I},N},\pi_N(\widehat{\mathbb{W}}_{j+1}^N\otimes e_2)\rangle(W_{j+1}-W_j).
		\end{flalign}
		The operator reads
		\begin{align*}
			\widehat{\mathcal{L}}_j^l
			=&\delta_x\widehat{u}_j^l\big[\widehat{\bs{q}}_j^{f,N}-f_0(t_j,x_l)\langle\nabla\mathcal{N}_j^{\mathbf{I},N},\pi_N(\widehat{\mathbb{W}}_j^N\otimes e_2)\rangle^2\big]\\
			&+\frac{1}{2}\delta^2_{xx}\widehat{u}_j^l\big[ g^2(t_j,x_l)(\mathcal{N}_j^{\mathbf{v},N})^2-f^2(t_j,x_l)\langle\nabla\mathcal{N}_j^{\mathbf{I},N},\pi_N(\widehat{\mathbb{W}}_j^N\otimes e_2)\rangle^2\big].
		\end{align*}
		}
		
		For comparison, we also introduce the neural network-based linear signature representations for these non-Markovian volatility processes.
		Since the coefficients of linear signature representations for complex volatility processes are implicit, we introduce another neural network $\mathcal{N}^{\bs{\ell},N}$ with parameters $\phi_{\bs{\ell}}$.
		It shares the same architecture as $\mathcal{N}^{\mathbf{v},N}$ but takes different outputs.
		Specifically, for each time step $j$, we have $\mathcal{N}^{\bs{\ell}}:\big(j,\widehat{\mathbb{W}}_j^N;\,\phi_{\bs{\ell}}\big)\mapsto\pi_N(\widehat{\bs{\ell}}_j)$ to learn the coefficients $\pi_N(\bs{\ell}_j)$ for $\lim\limits_{N\to\infty}\langle\pi_N(\bs{\ell}_j),\widehat{\mathbb{W}}_j^N\rangle=\mathbf{v}_j$.
		And the training problem can be defined as follows
		\begin{equation}
			\min\limits_{\phi_{\bs{\ell}}}\frac{1}{J}\sum\limits_{j=0}^J\Big(\big\langle\mathcal{N}^{\bs{\ell}}\big(j,\widehat{\mathbb{W}}_j^N;\,\phi_{\bs{\ell}}\big),\widehat{\mathbb{W}}_j^N\big\rangle-\mathbf{v}_j\Big)^2.
		\end{equation}
		\cpurple{
		Similarly, we define the neural network $\mathcal{N}^{\bs{p}}:\big(j,\widehat{\mathbb{W}}_j^N;\,\phi_{\bs{p}}\big)\mapsto\pi_N(\widehat{\bs{p}}_j)$, where $\pi_N(\widehat{\bs{p}}_j)$ is the neural network output to learn $\pi_N(\bs{p}_j)$. 
		}
		In the following example, we will employ these two neural networks to deal with option pricing problems under non-Markovian volatility processes.
		
		\begin{example}[Non-Markovian volatility processes with nonlinear and linear signature representations]
		\label{example-non-Markovian}
			\cpurple{
			In this example, we first use the previously introduced neural networks $(\mathcal{N}^{\mathbf{v},N},\mathcal{N}^{\mathbf{I},N})$ and $(\mathcal{N}^{\bs{\ell},N},\mathcal{N}^{\bs{p},N})$ to represent the rough pair $(\mathbf{v},\mathbf{I})$ of non-Markovian processes $(v,I)$, which denote the nonlinear and linear signature representations, respectively.}
			We then incorporate these signature representations into the corresponding SDE and PDE to compute European put option prices as in Example~\ref{example-Markovian}.
			
			For implementation details, we consider the Riemann-Liouville rHeston volatility process $v$, as well as the rBergomi process, from Table~\ref{tab-example-vol-processes}.
			We set the constants as follows
			\[
				\kappa=0.1,\;\theta=0.25,\;\sigma=0.01,\;\eta=1,\;v_0=0.1,\;\alpha=0.2,
\;\rho=-0.4,\;\beta=0.6,\;K=110,\;T=1.
			\]
			Moreover, we keep the other settings consistent with Example~\ref{example-Markovian}.
		\end{example}
			
			Let $(\widehat{\mathbf{v}}^{\text{nonlinear}},\widehat{\mathbf{I}}^{\text{nonlinear}})$ and $(\widehat{\mathbf{v}}^{\text{linear}},\widehat{\mathbf{I}}^{\text{linear}})$ denote the nonlinear and linear signature representations from 
            $(\mathcal{N}^{\mathbf{v},N},$\\
            $\mathcal{N}^{\mathbf{I},N})$ and $(\mathcal{N}^{\bs{\ell},N},\mathcal{N}^{\bs{p},N})$, respectively.
			Then we present the MAE $\mathcal{E}$ between the Monte Carlo benchmark $(\mathbf{v},\mathbf{I})$ and its truncated signature representations at different truncation levels for both rHeston and rBergomi processes in Figure~\ref{fig-example2-v-I-MAE}.
			The results show that: (1) \cpurple{while the nonlinear signature representation consistently outperforms its linear counterpart in accuracy across all truncation levels for the rBergomi volatility process, both representations show comparable performance under the rHeston model;} (2) the accuracy of both the nonlinear and linear signature representations increases significantly as the truncation level increases; (3) the accuracy of both the nonlinear and linear signature representations for the rHeston volatility process is significantly better than that for the rBergomi volatility process. However, as we will show next, this difference shrinks when evaluating option pricing errors, as illustrated in Example~\ref{example-Markovian}.

			\begin{figure}[H]
				\centering
				\subfigure[rHeston]{
					\includegraphics[width=0.49\textwidth]{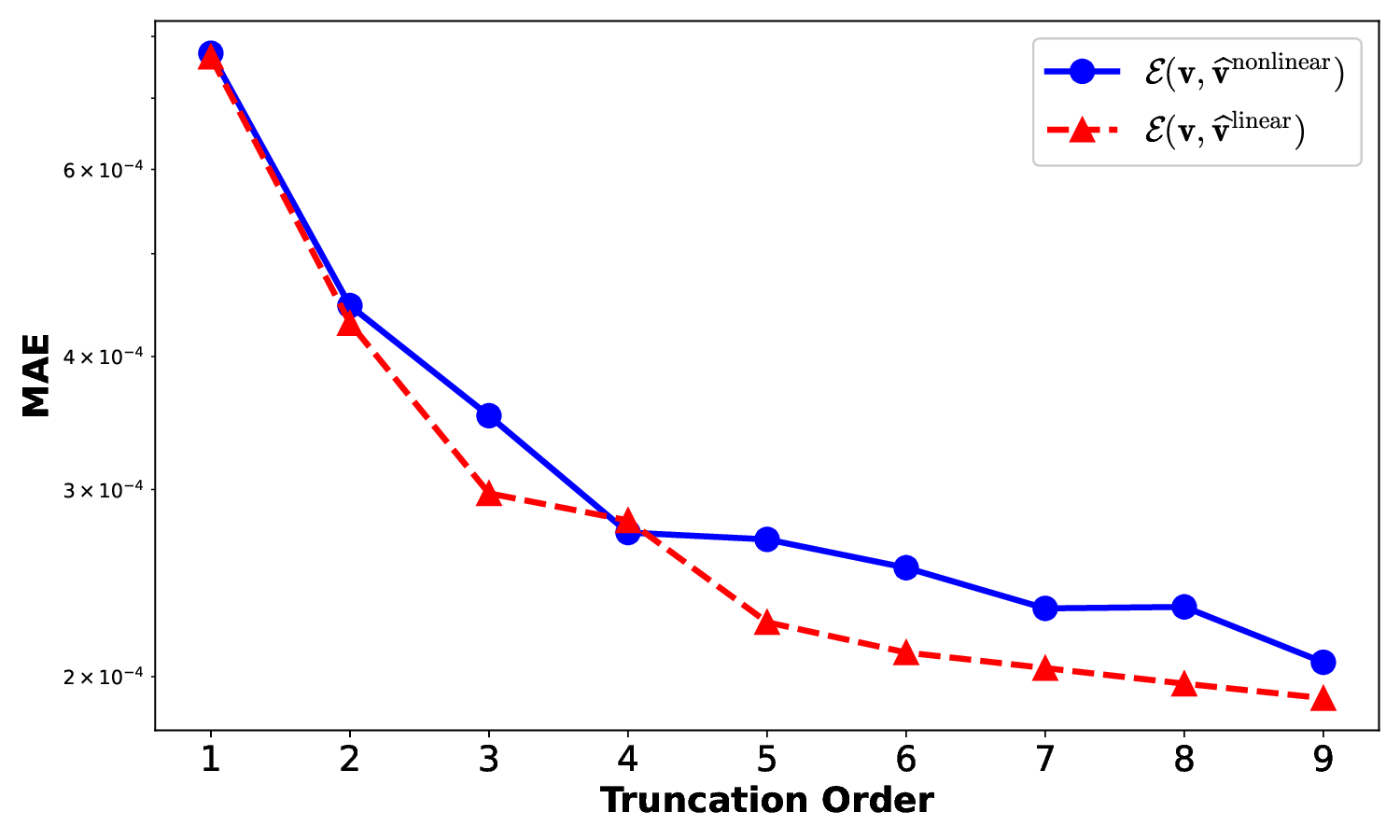}
					\includegraphics[width=0.49\textwidth]{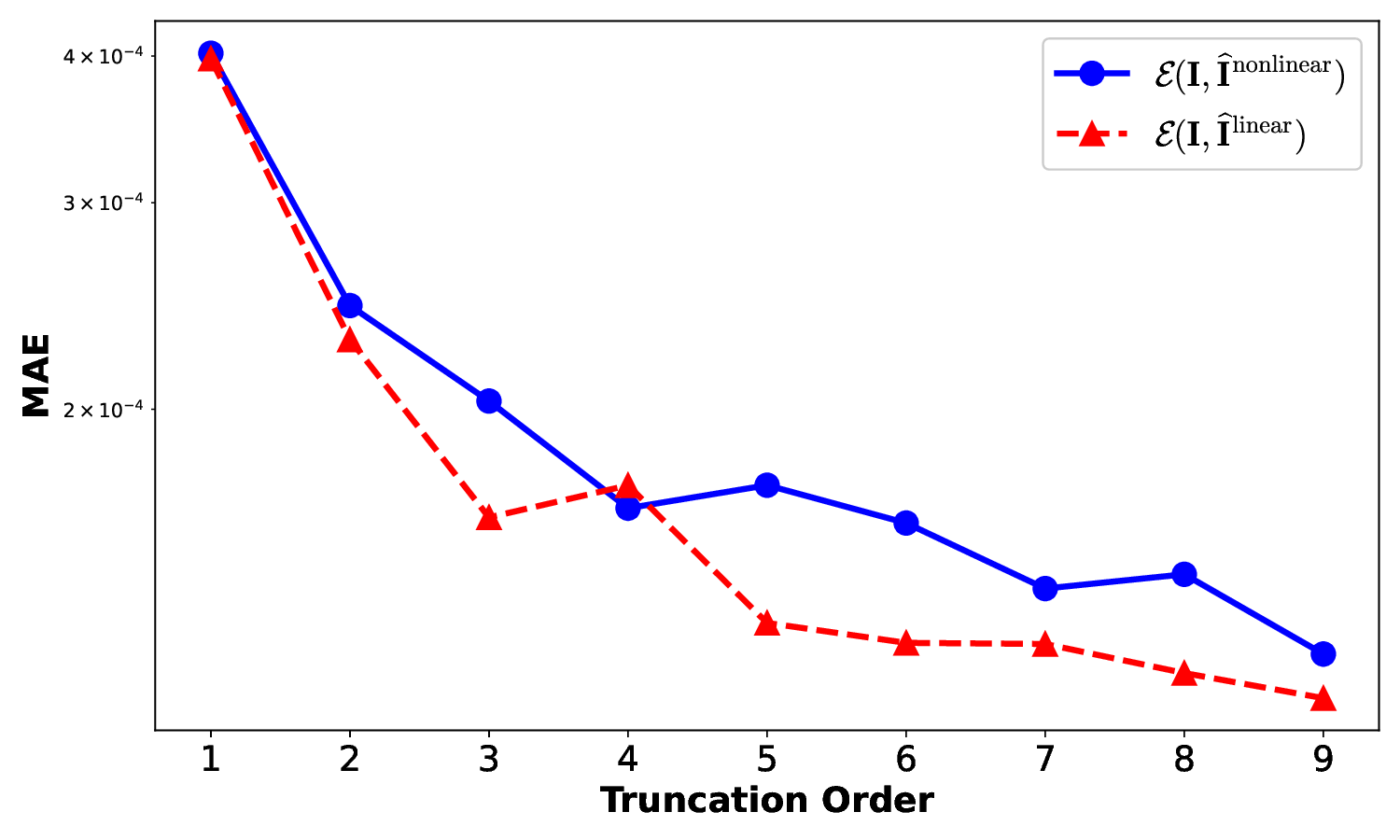}}
				\subfigure[rBergomi]{
					\includegraphics[width=0.49\textwidth]{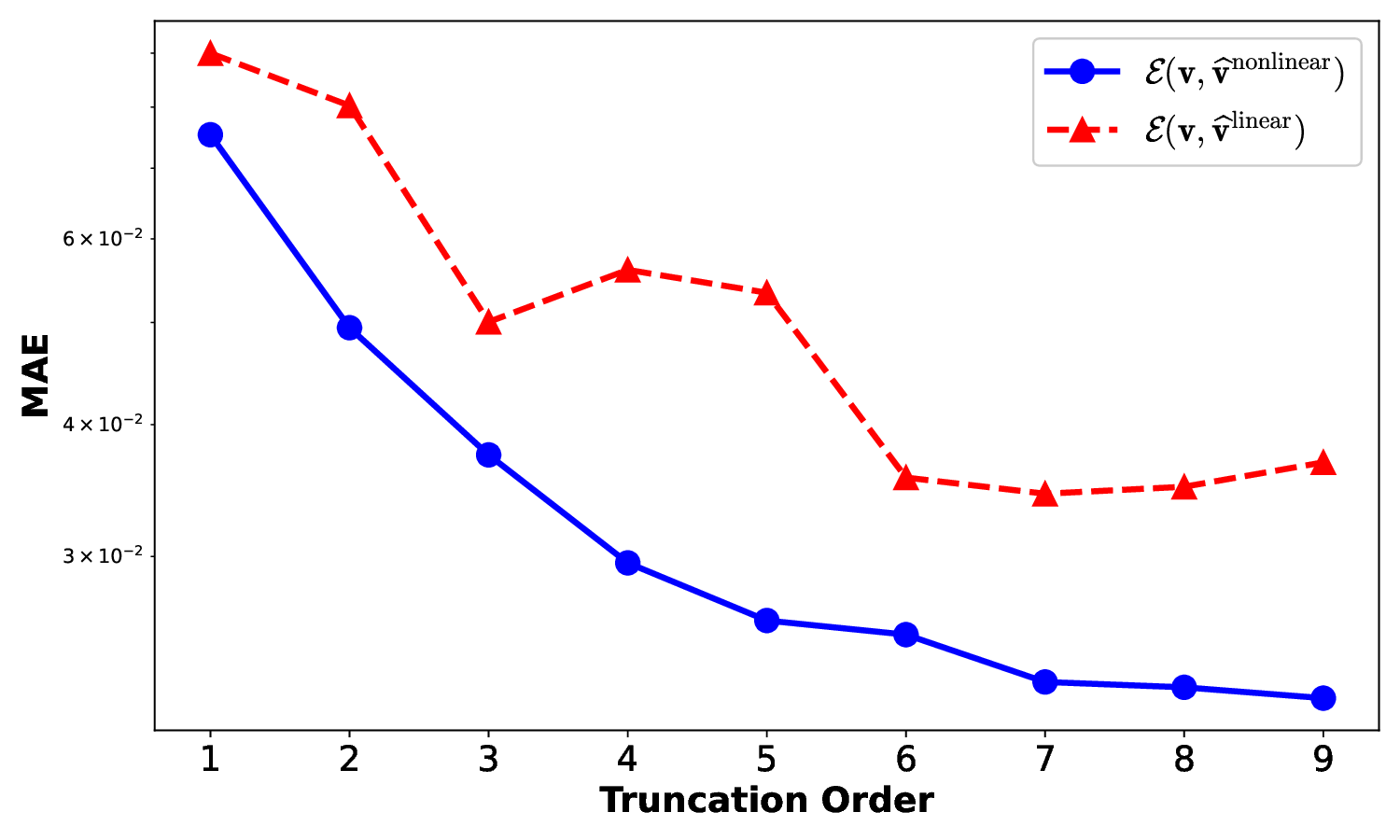}
					\includegraphics[width=0.49\textwidth]{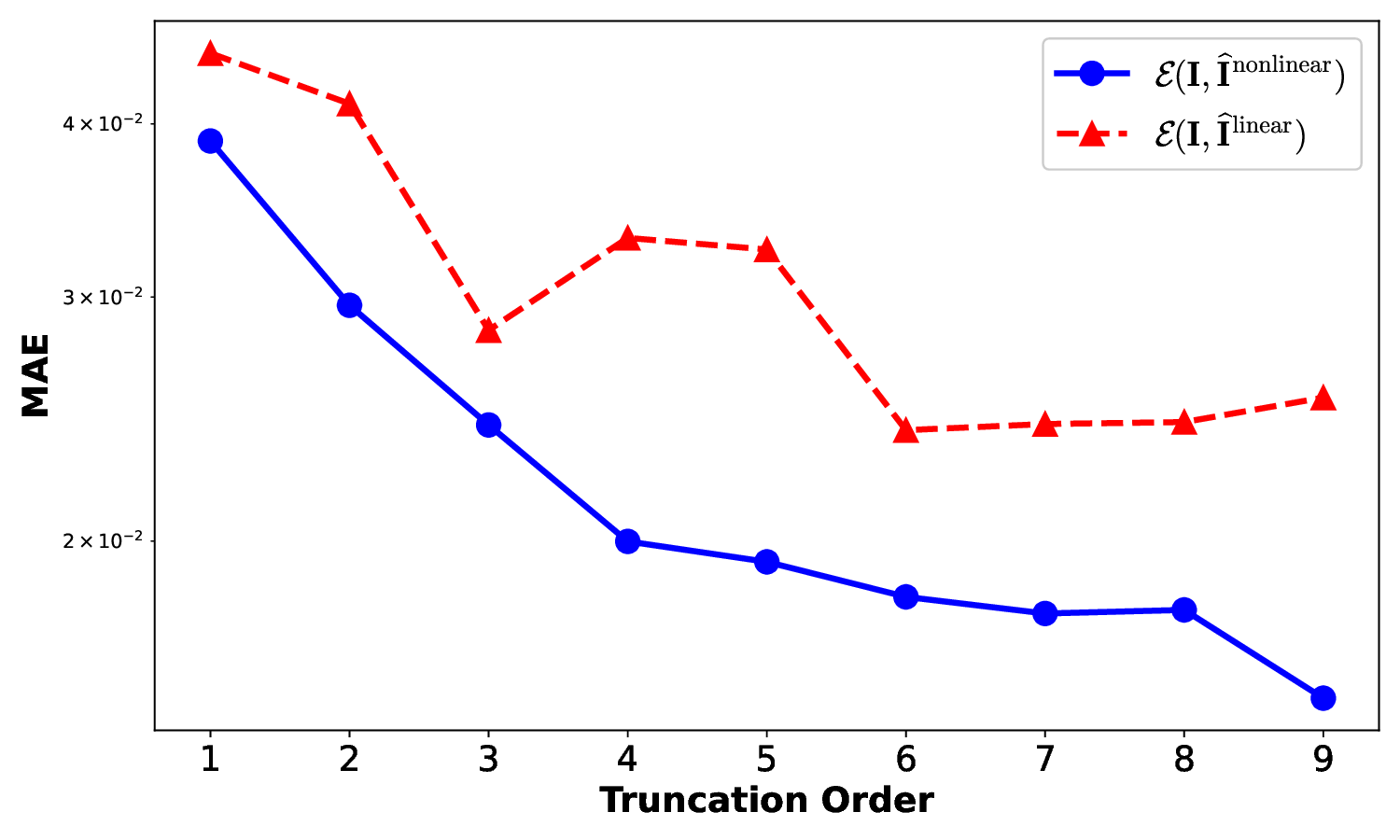}}
				\caption{Accuracy of nonlinear signature representations $(\widehat{\mathbf{v}}^{\text{nonlinear}},\widehat{\mathbf{I}}^{\text{nonlinear}})$ and linear signature representations $(\widehat{\mathbf{v}}^{\text{linear}},\widehat{\mathbf{I}}^{\text{linear}})$ for truncation levels $N=1$ to $9$ under the rHeston and rBergomi process. The left column are the MAEs for $\widehat{\mathbf{v}}^{\text{nonlinear}}$ and $\widehat{\mathbf{v}}^{\text{linear}}$, and the right column are the MAEs for $\widehat{\mathbf{I}}^{\text{nonlinear}}$ and $\widehat{\mathbf{I}}^{\text{linear}}$. The vertical axis uses a logarithmic scale and the results are based on $M=10,000$ samples.}
				\label{fig-example2-v-I-MAE}
			\end{figure}

			Similar to Example~\ref{example-Markovian}, we incorporate both the nonlinear and linear signature representations into the SDE and PDE for European put options pricing.
			\cpurple{
			Moreover, we use $\mathbb{E}\left[\Phi\left(\widehat{X}_T^{\text{nonlinear}}\right)\right]$ and $\mathbb{E}\left[\Phi\left(\widehat{X}_T^{\text{linear}}\right)\right]$ to denote the option prices obtained from the nonlinear and linear signature representations, respectively.
			}
			The results are shown in Figures~\ref{fig-example2-option-price-rHeston} and \ref{fig-example2-option-price-rBergomi} for the rHeston and rBergomi volatility processes, respectively. We observe that the approximation error converges quickly and that a small truncation level, such as $N=3$, is sufficient for practical purposes. The reason is that both the nonlinear and linear signature representations of the pair $(\mathbf{v}, \mathbf{I})$ have achieved sufficiently high accuracy at high truncation orders. Consequently, when these representations are incorporated into the SDE and PDE for option pricing, the errors remain comparable.

        	\begin{figure}[H]
        		\centering
        		\subfigure[SDE under rHeston]{
        			\includegraphics[width=0.29\textwidth]{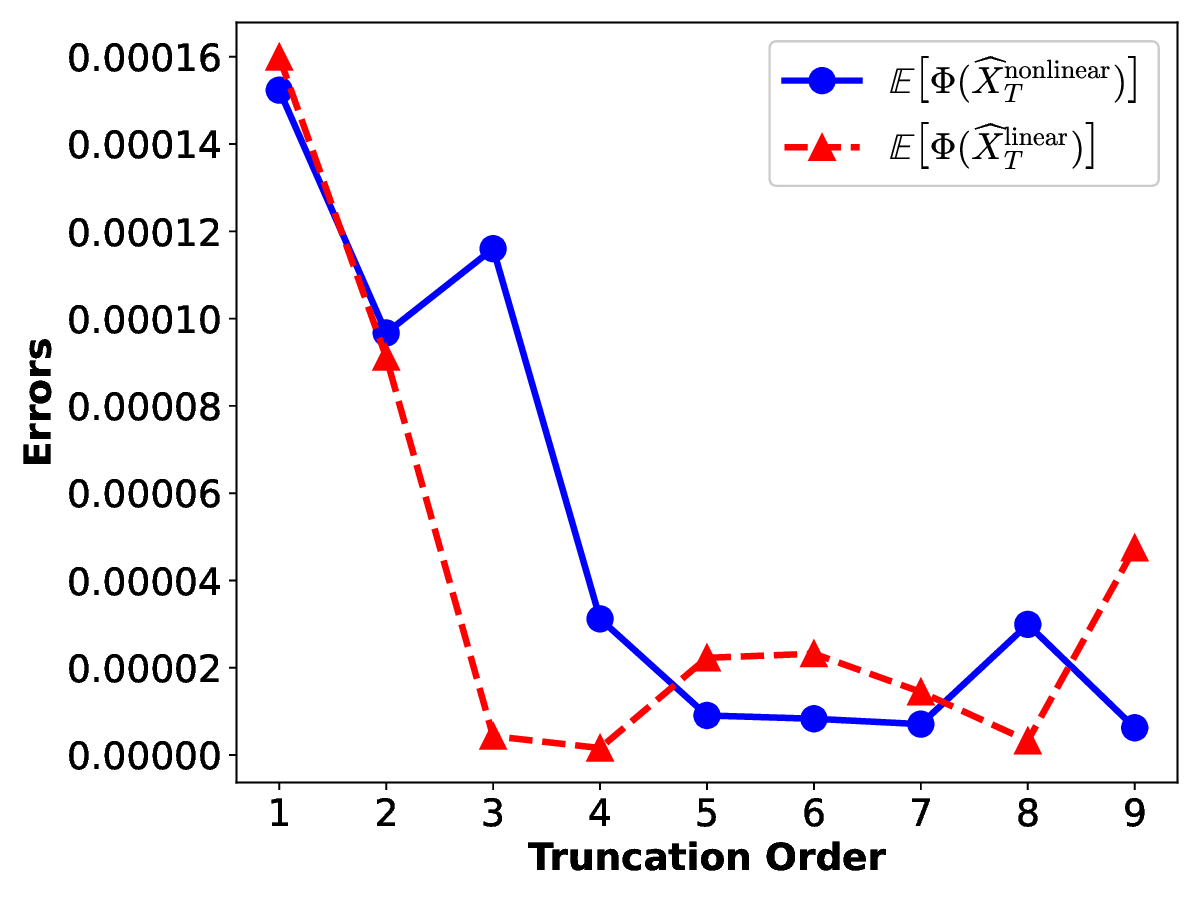}
        			\includegraphics[width=0.29\textwidth]{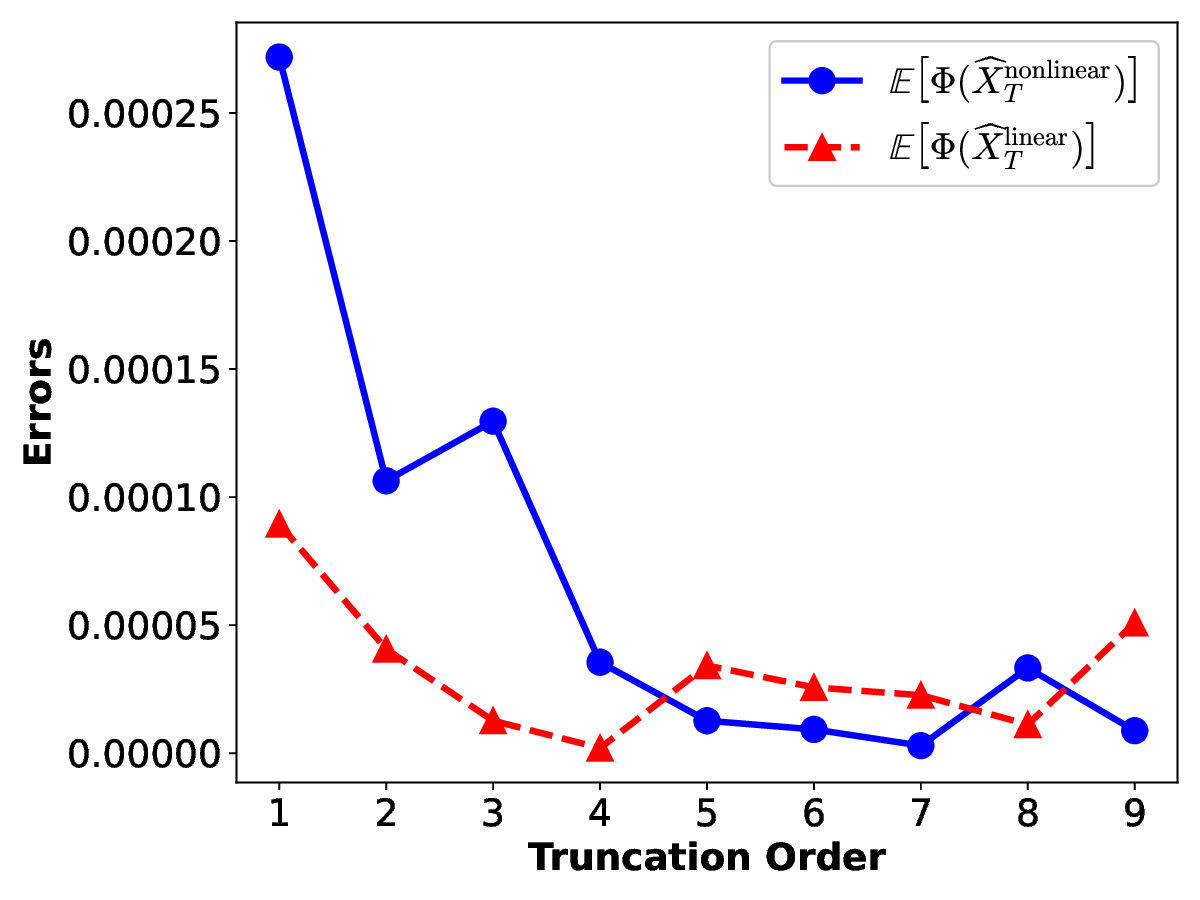}
        			\includegraphics[width=0.29\textwidth]{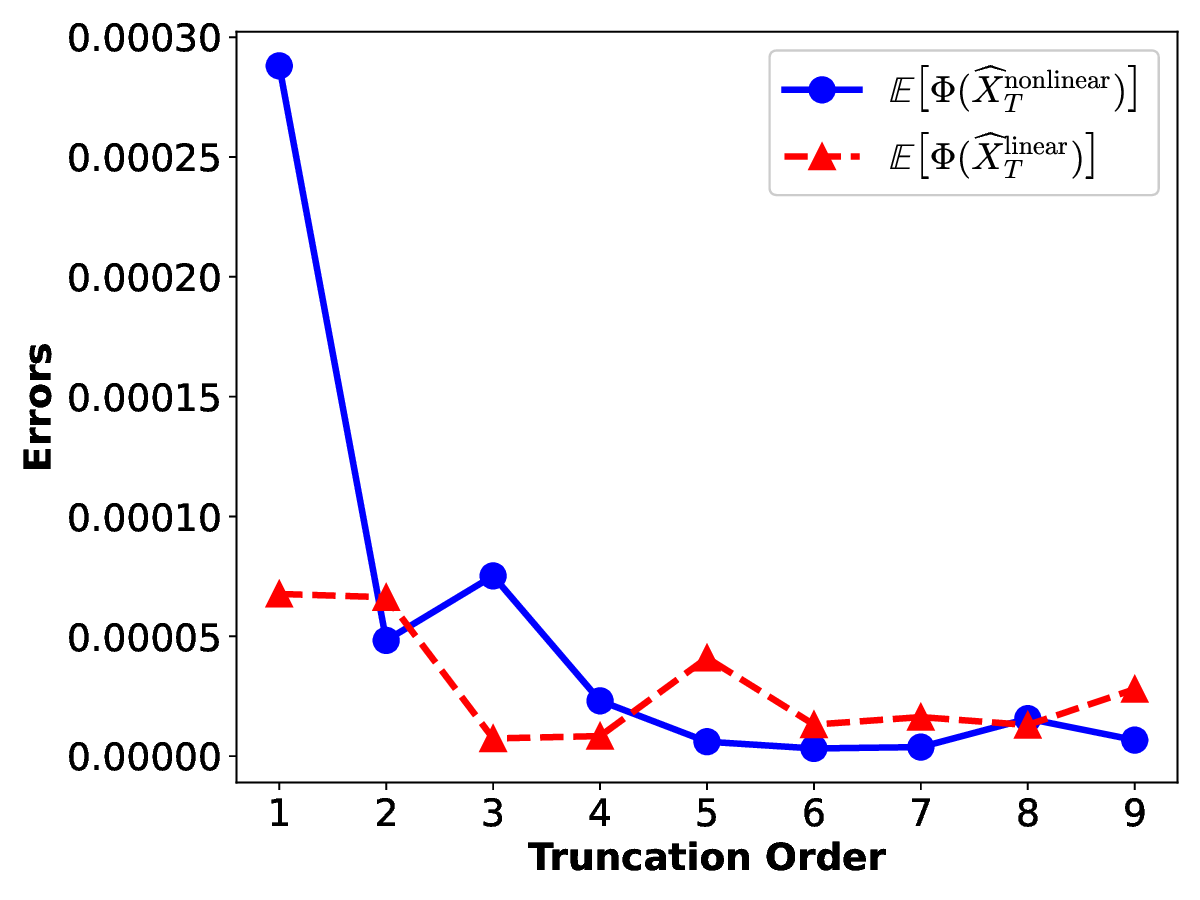}
        		}
        		\subfigure[PDE under rHeston]{
        			\includegraphics[width=0.29\textwidth]{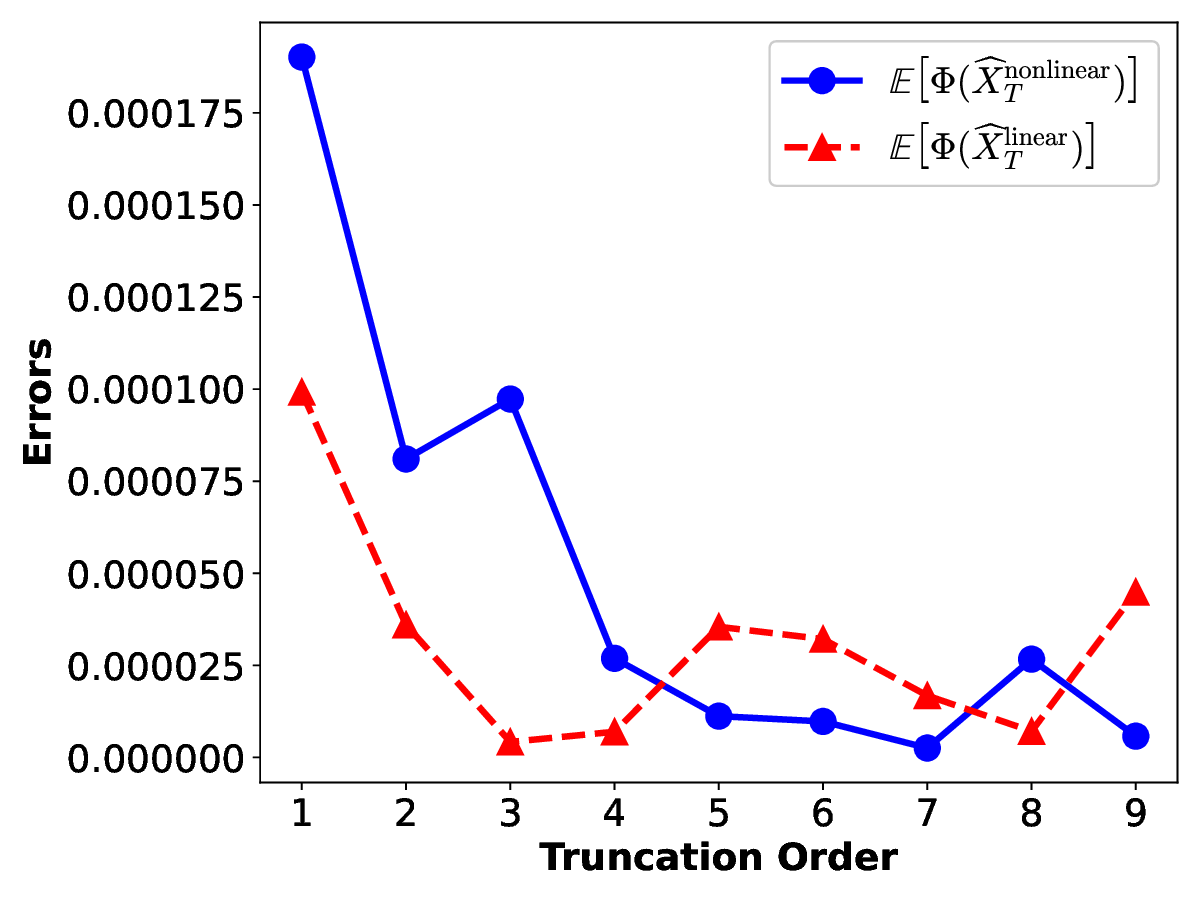}
        			\includegraphics[width=0.29\textwidth]{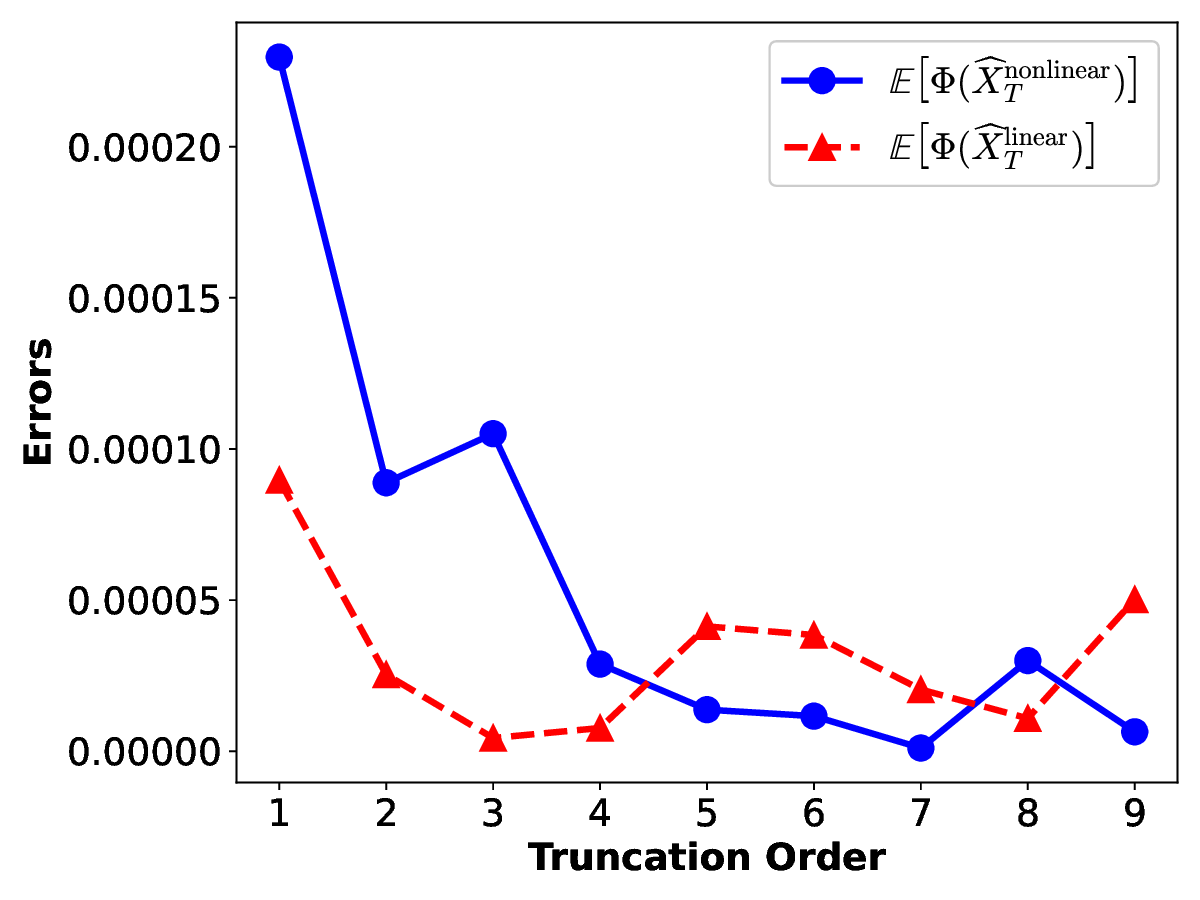}
        			\includegraphics[width=0.29\textwidth]{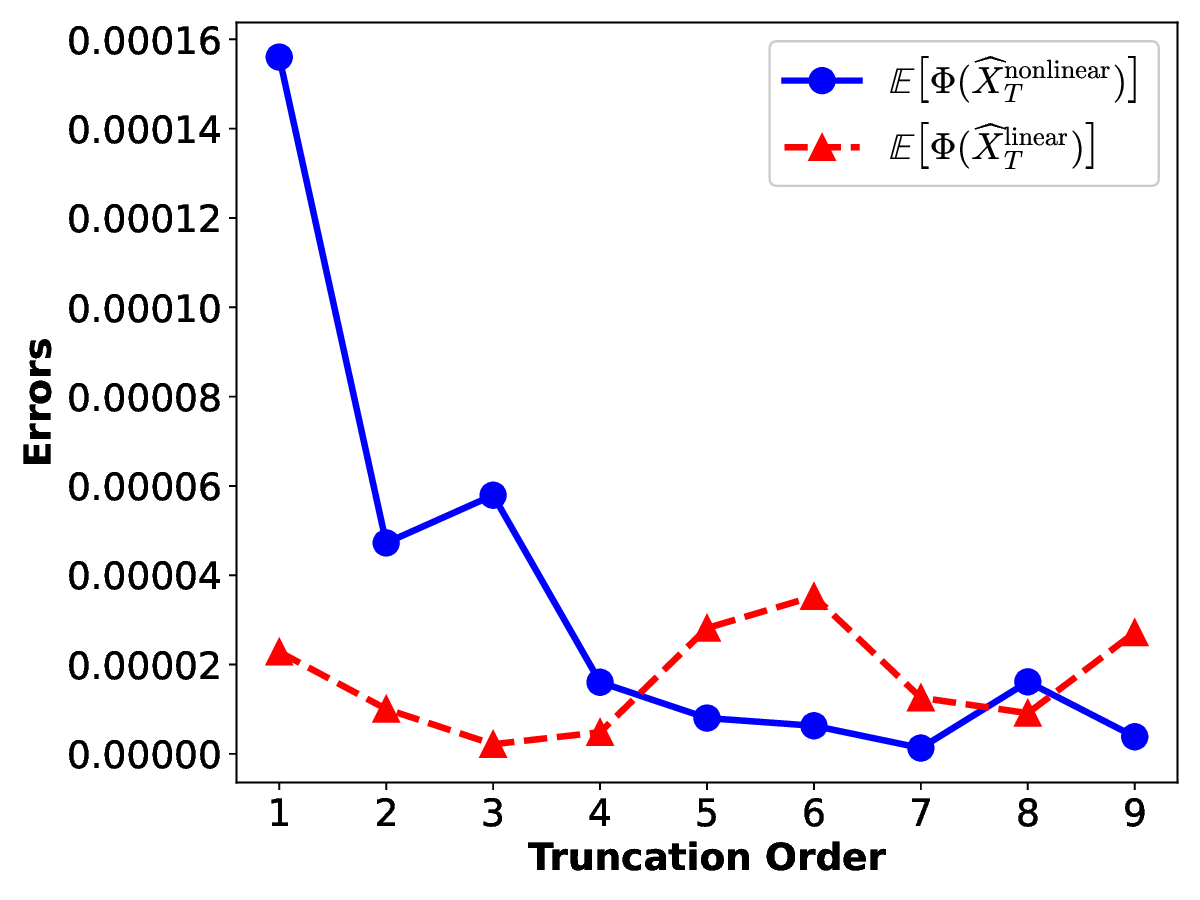}
        		}
        		\caption{Approximation error obtained from SDE and PDE for European put option prices under rHeston volatility processes at different moneyness levels. The left column is for OTM ($x_0=115>K$), the middle column is for ATM ($x_0=110=K$), and the right column is for ITM ($x_0=105<K$).}
        		\label{fig-example2-option-price-rHeston}
        	\end{figure}
    		
    		\begin{figure}[H]
    			\centering
    			\subfigure[SDE under rBergomi]{
    				\includegraphics[width=0.29\textwidth]{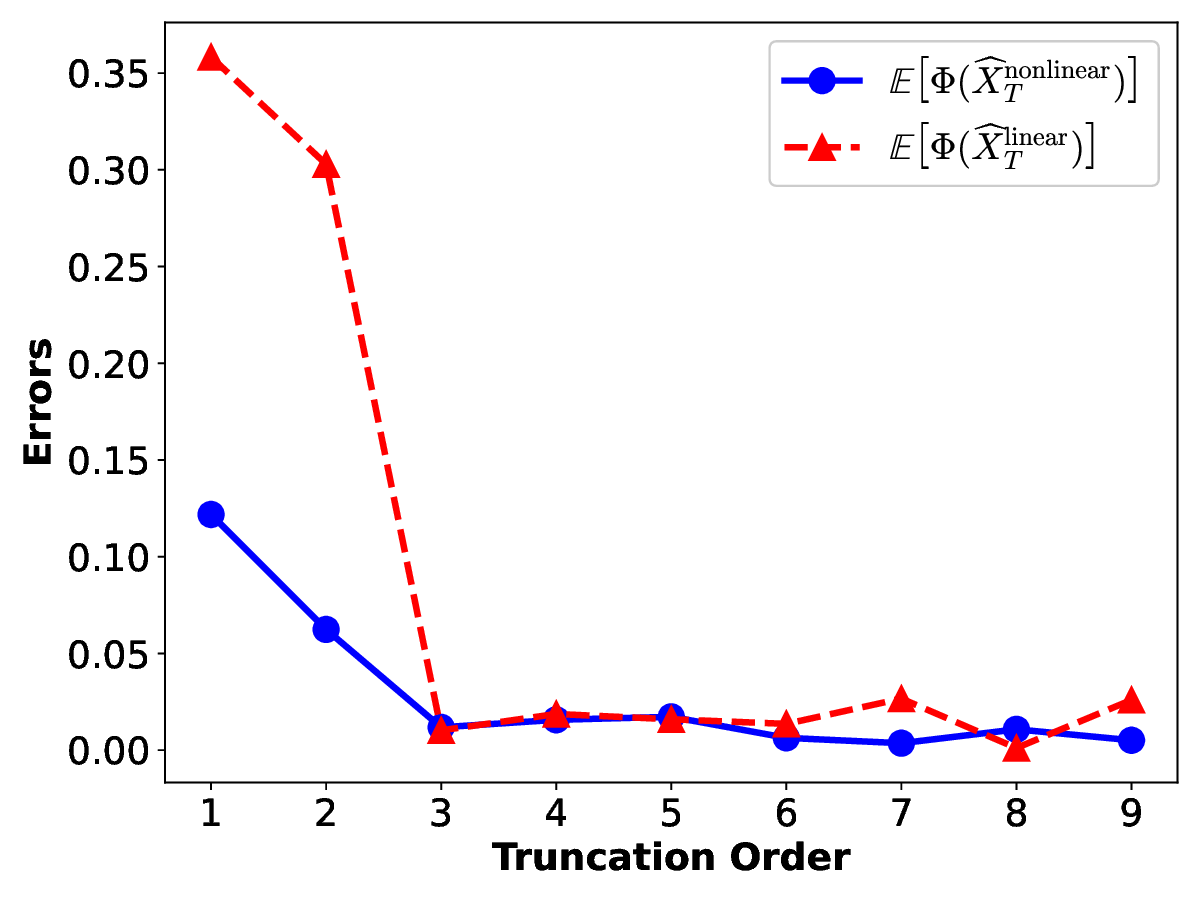}
    				\includegraphics[width=0.29\textwidth]{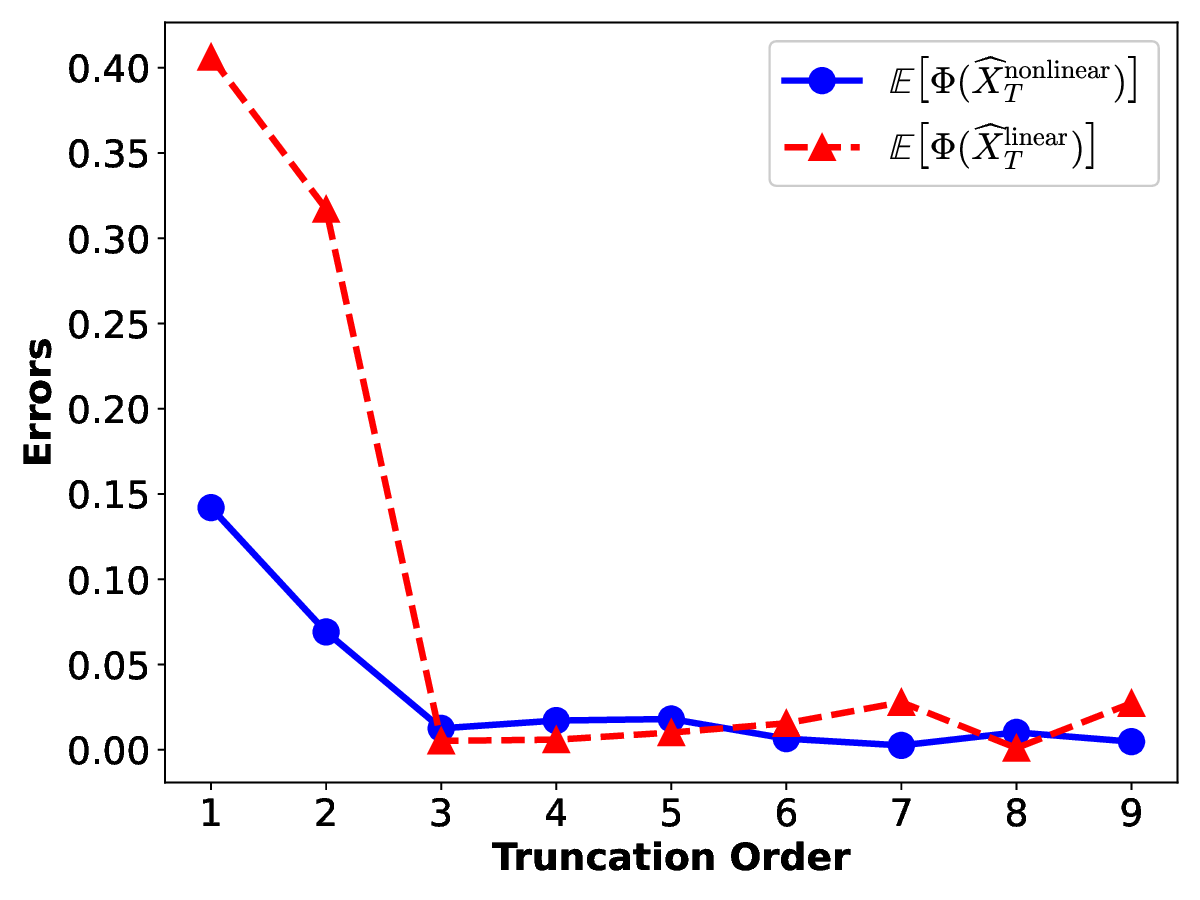}
    				\includegraphics[width=0.29\textwidth]{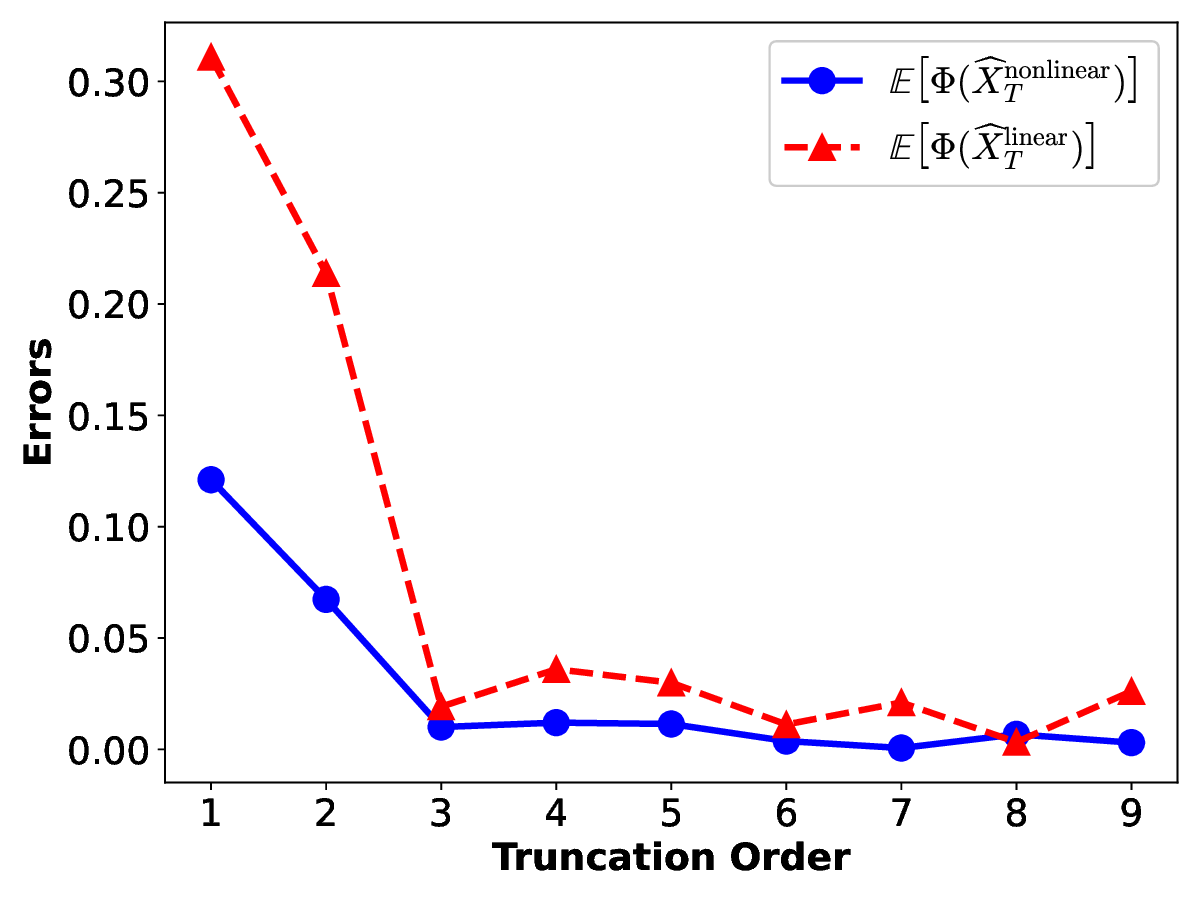}
    			}
    			\subfigure[PDE under rBergomi]{
    				\includegraphics[width=0.29\textwidth]{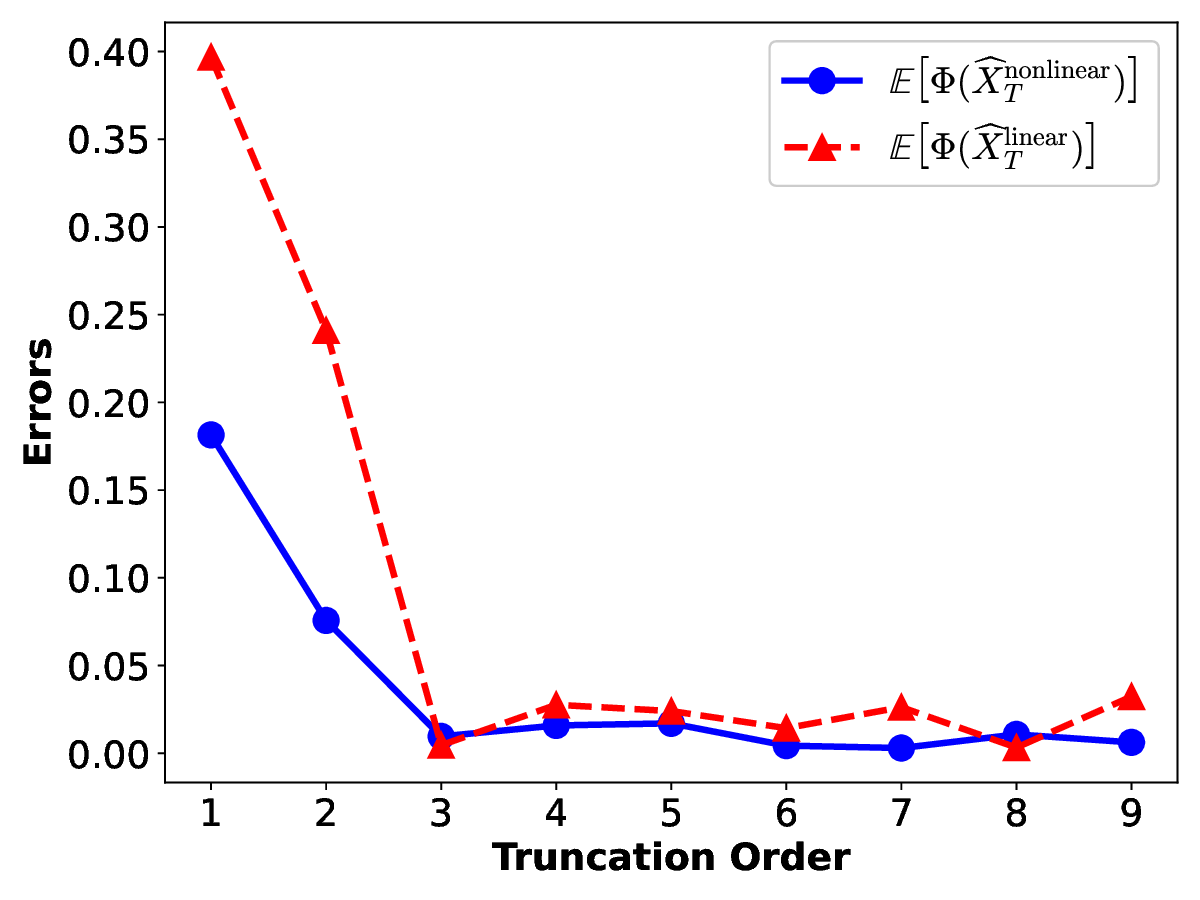}
    				\includegraphics[width=0.29\textwidth]{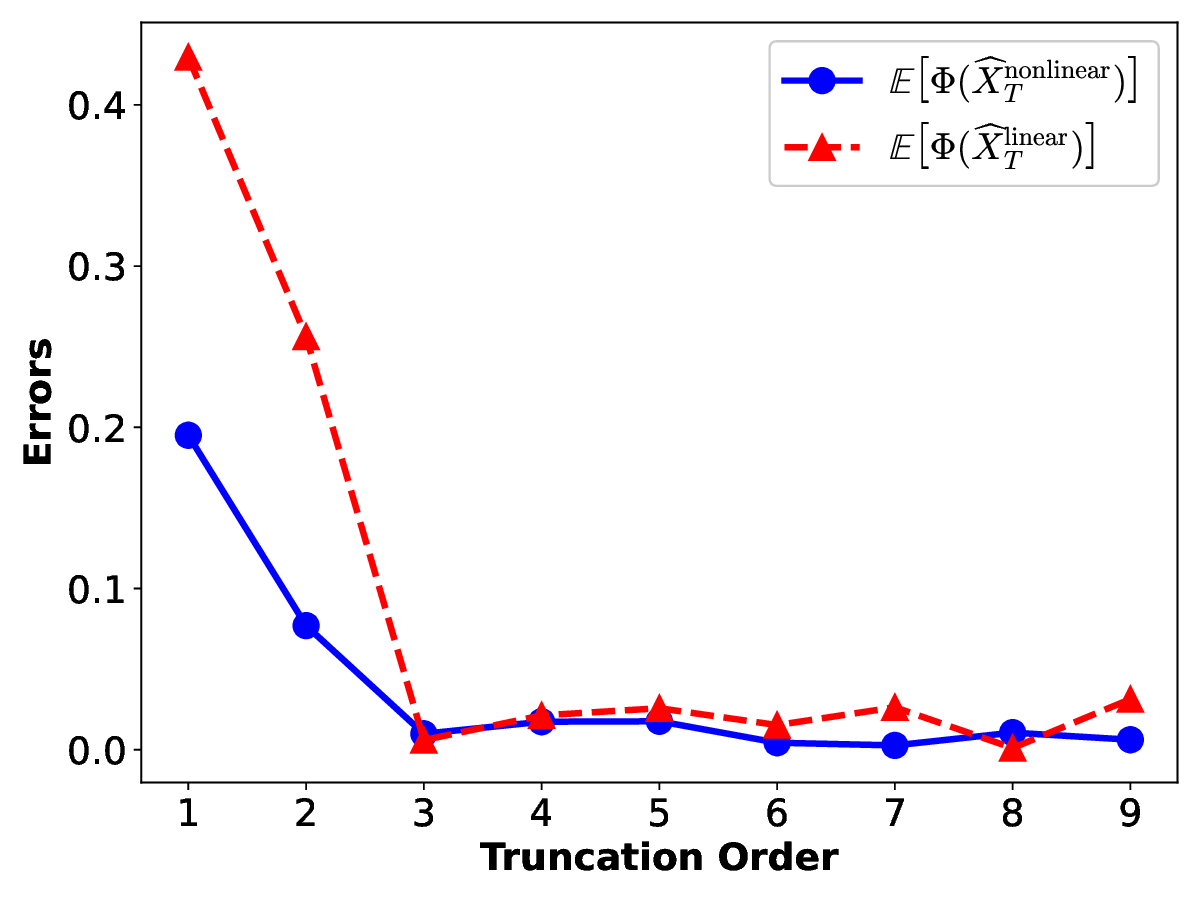}
    				\includegraphics[width=0.29\textwidth]{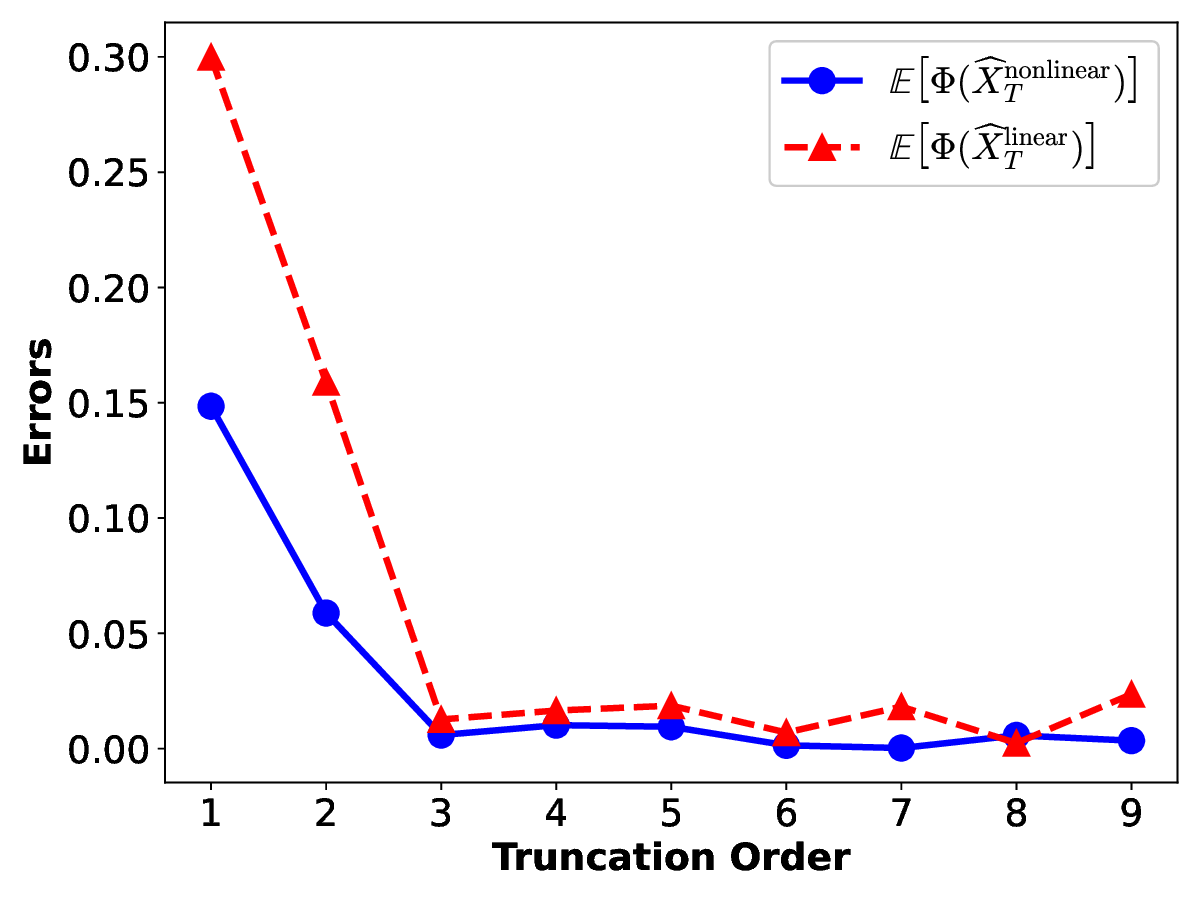}
    			}
    			\caption{Approximation error obtained from SDE and PDE for European put option prices under rBergomi volatility processes at different moneyness levels. The left column is for OTM ($x_0=115>K$), the middle column is for ATM ($x_0=110=K$), and the right column is for ITM ($x_0=105<K$).}
    			\label{fig-example2-option-price-rBergomi}
    		\end{figure}
            
    \section{Conclusions}
	\label{sec-conclusions}
	
	This paper proposes a deep signature method to solve the European option pricing problem within the framework of general non-Markovian stochastic volatility models. To address the non-Markovian nature of the volatility process, we first reformulate the SDE to an RSDE through Lyons' lift. Using the signature representation, we then transform the RSDE into a classical SDE. 
    \cpurple{By characterizing the value function as a random field, we derive the corresponding PDE for option pricing via the It\^o-Wentzell formula and the classical Feynman-Kac theorem. Subsequently, we employ the finite difference methods to compute the European option prices. }
    For linear Markovian models, such as OU and mGBM, we can obtain the explicit forms of the signature representations. For nonlinear path-dependent models, such as rHeston and rBergomi, we propose two approaches, deep linear signature and deep nonlinear signature methods, to study the non-Markovian volatility. Numerical experiments demonstrate that both methods converge rapidly with increasing signature truncation level, with the deep nonlinear signature method achieving superior accuracy. Furthermore, we provide a comprehensive convergence analysis that incorporates truncation errors of the signature, approximation errors of the signature representation, and the overall pricing error, thereby ensuring the robustness of our framework. Due to the powerful learning capacity of neural networks, our deep signature method is adaptable to more complex scenarios and opens avenues for future research in option pricing under non-Markovian stochastic volatility models.

\appendix
\enlargethispage{\baselineskip}
\section{Proofs}
    \subsection{Proof of Proposition \ref{proposition-X-gap}}\label{appendA2}
    	\begin{proof}
    		{\color{purple}
    		We first define the stopping times $\tau_m:=\inf\big\{s\in[0,t]:(|X_s|,|X_s^N|)\notin[2^{-m}, 2^m]^2\big\}$ and $\tau_k:=\inf\Big\{s\in[0,t]:\max\limits_{\bs{q}_s}\sum\limits_{n=0}^{\infty}|\langle\bs{q}_s^{(n)},\widehat{\mathbb{W}}_s^{(n)}\rangle|\geq k\text{ for }\bs{q}_s\in\{\bs{\ell}_s,\dot{\bs{p}}_s,\mathcal{D}_1\bs{p}_s,\mathcal{D}_2\bs{p}_s,\mathcal{D}_{22}^2\bs{p}_s\}\Big\}$ for $m,k\in\mathbb{N}^+$ and $t\in[0,T]$, where $\inf\{\emptyset\}=\infty$. 
            Note that the SABR model loses Lipschitz continuity at $0$ due to the unbounded derivative of the fractional-power terms in $f$ and $g$. 
    		By the definition of $\tau_m$, it is clear that $f$ and $g$ are Lipschitz continuous and satisfy the linear growth condition on $[0,t\land\tau_m]$.
    		We further define the joint stopping time $\tau_{m,k}:=\tau_m\land\tau_k$ and obtain the following estimate
    		\begin{flalign}
    			\label{eq-gap-X}
    			&\mathbb{E}\big|X_{t\land\tau_{m,k}}-X_{t\land\tau_{m,k}}^N\big|^2\\
    			\leq&3\mathbb{E}\Big[\Big|\int_0^{t\land\tau_{m,k}}\big(\langle\bs{q}^f_s,\widehat{\mathbb{W}}_s^{\infty}\rangle-\langle\pi_N(\bs{q}^f_s),\widehat{\mathbb{W}}_s^N\rangle\big)ds\Big|^2\nno\\
    			&+\Big|\int_0^{t\land\tau_{m,k}}\big(\langle f(s,X_s)\mathcal{D}_2\bs{p}_s,\widehat{\mathbb{W}}_s^{\infty}\rangle-\langle f(s,X_s^N)\pi_N(\mathcal{D}_2\bs{p}_s),\widehat{\mathbb{W}}_s^N\rangle\big)dW_s\Big|^2\nno\\
    			&+\Big|\int_0^{t\land\tau_{m,k}}\big(\langle g(s,X_s)\bs{\ell}_s,\widehat{\mathbb{W}}_s^{\infty}\rangle-\langle g(s,X_s^N)\pi_N(\bs{\ell}_s),\widehat{\mathbb{W}}_s^N\rangle\big)dB_s\Big|^2\Big]\nno\\
    			\leq&3\mathbb{E}\Big[T\int_0^{t\land\tau_{m,k}}\big|\langle\bs{q}^f_s,\widehat{\mathbb{W}}_s^{\infty}\rangle-\langle\pi_N(\bs{q}^f_s),\widehat{\mathbb{W}}_s^N\rangle\big|^2ds\nno\\
    			&+\int_0^{t\land\tau_{m,k}}\big|\langle f(s,X_s)\mathcal{D}_2\bs{p}_s,\widehat{\mathbb{W}}_s^{\infty}\rangle-\langle f(s,X_s^N)\pi_N(\mathcal{D}_2\bs{p}_s),\widehat{\mathbb{W}}_s^N\rangle\big|^2ds\nno\\
    			&+\int_0^{t\land\tau_{m,k}}\big|\langle g(s,X_s)\bs{\ell}_s,\widehat{\mathbb{W}}_s^{\infty}\rangle-\langle g(s,X_s^N)\pi_N(\bs{\ell}_s),\widehat{\mathbb{W}}_s^N\rangle\big|^2ds\Big],\nno
    		\end{flalign}
    		where the second inequality uses the It\^o isometry and Cauchy-Schwarz inequality.
    		This decomposition allows us to analyze the gap between $X_t$ and $X_t^N$ separately.
    		
    		When $s\in[0,t\land\tau_{m,k}]$, for a given path in the complete filtered probability space $(\Omega,\mathcal{F}_t,P)$, we have
    		\begin{flalign}
    			\label{eq-gap-qf}
    			&\big|\langle\bs{q}^f_s,\widehat{\mathbb{W}}_s^{\infty}\rangle-\langle\pi_N(\bs{q}^f_s),\widehat{\mathbb{W}}_s^N\rangle\big|^2\\
    			=&\big|f(s,X_s)\langle \dot{\bs{p}}_s+\mathcal{D}_1\bs{p}_s+\frac{1}{2}\mathcal{D}_{22}^2\bs{p}_s,\widehat{\mathbb{W}}_s^{\infty}\rangle\nno\\
    			&-f(s,X_s^N)\langle \pi_N(\dot{\bs{p}}_s)+\pi_N(\mathcal{D}_1\bs{p}_s)+\frac{1}{2}\pi_N(\mathcal{D}_{22}^2\bs{p}_s),\widehat{\mathbb{W}}_s^N\rangle\big|^2\nno\\
    			\leq&3\Big[\big|f(s,X_s)\langle\dot{\bs{p}}_s,\widehat{\mathbb{W}}_s^{\infty}\rangle-f(s,X_s^N)\langle\pi_N(\dot{\bs{p}}_s),\widehat{\mathbb{W}}_s^N\rangle\big|^2\nno\\
    			&+\big|f(s,X_s)\langle\mathcal{D}_1\bs{p}_s,\widehat{\mathbb{W}}_s^{\infty}\rangle-f(s,X_s^N)\langle\pi_N(\mathcal{D}_1\bs{p}_s),\widehat{\mathbb{W}}_s^N\rangle\big|^2\nno\\
    			&+\frac{1}{4}\big|f(s,X_s)\langle\mathcal{D}_{22}^2\bs{p}_s,\widehat{\mathbb{W}}_s^{\infty}\rangle-f(s,X_s^N)\langle\pi_N(\mathcal{D}_{22}^2\bs{p}_s),\widehat{\mathbb{W}}_s^N\rangle\big|^2\Big].\nno
    		\end{flalign}
    		The first part of \eqref{eq-gap-qf} can be treated as follows
    		\begin{align}
    			\label{eq-gap-qf-1}
    			&\big|f(s,X_s)\langle\dot{\bs{p}}_s,\widehat{\mathbb{W}}_s^{\infty}\rangle-f(s,X_s^N)\langle\pi_N(\dot{\bs{p}}_s),\widehat{\mathbb{W}}_s^N\rangle\big|^2\\
    			\leq&2\Big[\big|f(s,X_s)\langle\dot{\bs{p}}_s,\widehat{\mathbb{W}}_s^{\infty}\rangle-f(s,X_s^N)\langle\dot{\bs{p}}_s,\widehat{\mathbb{W}}_s^{\infty}\rangle\big|^2\nno\\
    			+&\big|f(s,X_s^N)\langle\dot{\bs{p}}_s,\widehat{\mathbb{W}}_s^{\infty}\rangle-f(s,X_s^N)\langle\pi_N(\dot{\bs{p}}_s),\widehat{\mathbb{W}}_s^N\rangle\big|^2\Big]\nno\\
    			\leq&2\big(L^2_f\big|X_s-X_s^N\big|^2\langle\dot{\bs{p}}_s\shuffle\dot{\bs{p}}_s,\widehat{\mathbb{W}}_s^{\infty}\rangle+f^2(s,X_s^N)G_s^{\dot{\bs{p}},N,\omega}\big).\nno
    		\end{align}
    		where $G_s^{\dot{\bs{p}},N,\omega}:=\big|\langle\dot{\bs{p}}_s,\widehat{\mathbb{W}}_s^{\infty}\rangle-\langle\pi_N(\dot{\bs{p}}_s),\widehat{\mathbb{W}}_s^N\rangle\big|^2$ and $L_f\geq0$ is the Lipschitz constant for $f$.
    		The rest two parts of~\eqref{eq-gap-qf} can be treated similarly as follows
    		\begin{align}
    			\label{eq-gap-qf-2}
    			&\big|f(s,X_s)\langle\mathcal{D}_1\bs{p}_s,\widehat{\mathbb{W}}_s^{\infty}\rangle-f(s,X_s^N)\langle\pi_N(\mathcal{D}_1\bs{p}_s),\widehat{\mathbb{W}}_s^N\rangle\big|^2\\
    			\leq&2\big(L^2_f\big|X_s-X_s^N\big|^2\langle\mathcal{D}_1\bs{p}_s\shuffle\mathcal{D}_1\bs{p}_s,\widehat{\mathbb{W}}_s^{\infty}\rangle+f^2(s,X_s^N)G_s^{\mathcal{D}_1\bs{p},N,\omega}\big),\nno
    		\end{align}
    		\begin{align}
    			\label{eq-gap-qf-3}
    			&\big|f(s,X_s)\langle\mathcal{D}_{22}^2\bs{p}_s,\widehat{\mathbb{W}}_s^{\infty}\rangle-f(s,X_s^N)\langle\pi_N(\mathcal{D}_{22}^2\bs{p}_s),\widehat{\mathbb{W}}_s^N\rangle\big|^2\\
    			\leq&2\big(L^2_f\big|X_s-X_s^N\big|^2\langle\mathcal{D}_{22}^2\bs{p}_s\shuffle\mathcal{D}_{22}^2\bs{p}_s,\widehat{\mathbb{W}}_s^{\infty}\rangle+f^2(s,X_s^N)G_s^{\mathcal{D}_{22}^2\bs{p},N,\omega}\big),\nno
    		\end{align}
    		where
    		$G_s^{\mathcal{D}_1\bs{p},N,\omega}:=\big|\langle\mathcal{D}_1\bs{p}_s,\widehat{\mathbb{W}}_s^{\infty}\rangle-\langle\pi_N(\mathcal{D}_1\bs{p}_s),\widehat{\mathbb{W}}_s^N\rangle\big|^2$ and 
    		$G_s^{\mathcal{D}_{22}^2\bs{p},N,\omega}:=\big|\langle\mathcal{D}_{22}^2\bs{p}_s,\widehat{\mathbb{W}}_s^{\infty}\rangle-\langle\pi_N(\mathcal{D}_{22}^2\bs{p}_s),\widehat{\mathbb{W}}_s^N\rangle\big|^2$. 
    		Combing \eqref{eq-gap-qf}-\eqref{eq-gap-qf-3}, we have
    		\begin{align}
    			\label{eq-gap-qf-final}
    			&\big|\langle\bs{q}^f_s,\widehat{\mathbb{W}}_s^{\infty}\rangle-\langle\pi_N(\bs{q}^f_s),\widehat{\mathbb{W}}_s^N\rangle\big|^2\\
    			\leq&6L^2_f\big|X_s-X_s^N\big|^2\langle\dot{\bs{p}}_s\shuffle\dot{\bs{p}}_s+\mathcal{D}_1\bs{p}_s\shuffle\mathcal{D}_1\bs{p}_s+\frac{1}{4}\mathcal{D}_{22}^2\bs{p}_s\shuffle\mathcal{D}_{22}^2\bs{p}_s,\widehat{\mathbb{W}}_s^{\infty}\rangle\nno\\
    			&+6f^2(s,X_s^N)(G_s^{\dot{\bs{p}},N,\omega}+G_s^{\mathcal{D}_1\bs{p},N,\omega}+\frac{1}{4}G_s^{\mathcal{D}_{22}^2\bs{p},N,\omega}).\nno
    		\end{align}
    		
    		Similarly, the second and third parts of~\eqref{eq-gap-X} can be treated as follows
    		\begin{flalign}
    			\label{eq-gap-X-2}
    			&\big|\langle f(s,X_s)\mathcal{D}_2\bs{p}_s,\widehat{\mathbb{W}}_s^{\infty}\rangle-\langle f(s,X_s^N)\pi_N(\mathcal{D}_2\bs{p}_s),\widehat{\mathbb{W}}_s^N\rangle\big|^2\\
    			\leq&2\big(L^2_f\big|X_s-X_s^N\big|^2\langle\mathcal{D}_2\bs{p}_s\shuffle\mathcal{D}_2\bs{p}_s,\widehat{\mathbb{W}}_s^{\infty}\rangle+f^2(s,X_s^N)G_s^{\mathcal{D}_2\bs{p},N,\omega}\big),\nno
    		\end{flalign}
    		\begin{flalign}
    			\label{eq-gap-X-3}
    			&\big|\langle g(s,X_s)\bs{\ell}_s,\widehat{\mathbb{W}}_s^{\infty}\rangle-\langle g(s,X_s^N)\pi_N(\bs{\ell}_s),\widehat{\mathbb{W}}_s^N\rangle\big|^2\\
    			\leq&2\big(L^2_g\big|X_s-X_s^N\big|^2\langle\bs{\ell}_s\shuffle\bs{\ell}_s,\widehat{\mathbb{W}}_s^{\infty}\rangle+g^2(s,X_s^N)G_s^{\bs{\ell},N,\omega}\big),\nno
    		\end{flalign}
    		where $G_s^{\bs{\ell},N,\omega}:=\big|\langle\bs{\ell}_s,\widehat{\mathbb{W}}_s^{\infty}\rangle-\langle\pi_N(\bs{\ell}_s),\widehat{\mathbb{W}}_s^N\rangle\big|^2$, 
    		$G_s^{\mathcal{D}_2\bs{p},N,\omega}:=\big|\langle\mathcal{D}_2\bs{p}_s,\widehat{\mathbb{W}}_s^{\infty}\rangle-\langle\pi_N(\mathcal{D}_2\bs{p}_s),\widehat{\mathbb{W}}_s^N\rangle\big|^2$, and $L_g\geq0$ is the Lipschitz constant for $g$.

    		Combining \eqref{eq-gap-X}, \eqref{eq-gap-qf-final}, \eqref{eq-gap-X-2}, and \eqref{eq-gap-X-3}, we obtain
    		\begin{flalign*}
    			&\mathbb{E}\big|X_{t\land\tau_{m,k}}-X_{t\land\tau_{m,k}}^N\big|^2\\
				\leq&3T\mathbb{E}\Big\{\int_0^{t\land\tau_{m,k}}\Big[\big|X_s-X_s^N\big|^2\big(2L_g^2\langle\bs{\ell}_s\shuffle\bs{\ell}_s,\widehat{\mathbb{W}}_s^{\infty}\rangle\\
				&+6L_f^2\langle\dot{\bs{p}}_s\shuffle\dot{\bs{p}}_s+\mathcal{D}_1\bs{p}_s\shuffle\mathcal{D}_1\bs{p}_s+\frac{1}{4}\mathcal{D}_{22}^2\bs{p}_s\shuffle\mathcal{D}_{22}^2\bs{p}_s+\frac{1}{3}\mathcal{D}_2\bs{p}_s\shuffle\mathcal{D}_2\bs{p}_s,\widehat{\mathbb{W}}_s^{\infty}\rangle\big)\\
				&+6f^2(s,X_s^N)\big(G_s^{\dot{\bs{p}},N,\omega}+G_s^{\mathcal{D}_1\bs{p},N,\omega}+\frac{1}{4}G_s^{\mathcal{D}_{22}^2\bs{p},N,\omega}+\frac{1}{3}G_s^{\mathcal{D}_2\bs{p},N,\omega}\big)+2g^2(s,X_s^N)G_s^{\bs{\ell},N,\omega}\Big]ds\Big\}\\
				=&\mathbb{E}\big[F_1(t\land\tau_{m,k},\omega)\big]+\mathbb{E}\Big[\int_0^{t\land\tau_{m,k}}F_2(s,\omega)\big|X_s-X_s^N\big|^2ds\Big],
    		\end{flalign*}
    		where
    		\begin{flalign*}
    			F_1(t\land\tau_{m,k},\omega)
    			:=&3T\int_0^{t\land\tau_{m,k}}\Big[6f^2(s,X_s^N)\big(G_s^{\dot{\bs{p}},N,\omega}+G_s^{\mathcal{D}_1\bs{p},N,\omega}+\frac{1}{4}G_s^{\mathcal{D}_{22}^2\bs{p},N,\omega}+\frac{1}{3}G_s^{\mathcal{D}_2\bs{p},N,\omega}\big)\\
    			&+2g^2(s,X_s^N)G_s^{\bs{\ell},N,\omega}\Big]ds,
    		\end{flalign*}
    		and
    		\begin{flalign*}
    			F_2(s,\omega)
    			:=&3T\big(2L_g^2\langle\bs{\ell}_s\shuffle\bs{\ell}_s,\widehat{\mathbb{W}}_s^{\infty}\rangle\\
    			&+6L_f^2\langle\dot{\bs{p}}_s\shuffle\dot{\bs{p}}_s+\mathcal{D}_1\bs{p}_s\shuffle\mathcal{D}_1\bs{p}_s+\frac{1}{4}\mathcal{D}_{22}^2\bs{p}_s\shuffle\mathcal{D}_{22}^2\bs{p}_s+\frac{1}{3}\mathcal{D}_2\bs{p}_s\shuffle\mathcal{D}_2\bs{p}_s,\widehat{\mathbb{W}}_s^{\infty}\rangle\big),
    		\end{flalign*}
    		for $s\in[0,t\land\tau_{m,k}]$.

    		Given the linear growth conditions of $f$ and $g$, along with the definition of $\tau_m$, it follows that $f^2(s,X_s^N)$ and $g^2(s,X_s^N)$ are bounded and finitely integrable on $s\in[0,t\land\tau_m]$. 
    		Theorem~\ref{theorem-linear-rep-gap} implies that the expectations of $G_s^{\dot{\bs{p}},N,\omega},G_s^{\mathcal{D}_1\bs{p},N,\omega},G_s^{\mathcal{D}_{22}^2\bs{p},N,\omega},G_s^{\mathcal{D}_2\bs{p},N,\omega},G_s^{\bs{\ell},N,\omega}$ tend to zero as $N\to\infty$. 
    		Furthermore, by the definition of $\tau_k$, the linear combinations of signatures in $F_2(s,\omega)$ are bounded on $s\in[0,t\land\tau_k]$. 
    		These suggest that there exist nonnegative constants $C_{m,k}^N$ and $C_k$, which depend on $m$, $k$ and the truncation order $N$, such that
    		\begin{equation*}
    			\mathbb{E}\big[F_1(t\land\tau_{m,k},\omega)\big]=C_{m,k}^N,\quad F_2(s,\omega)\leq C_k,\, s\in[0,{t\land\tau_{m,k}}].
    		\end{equation*}
    		Therefore, using Gronwall's inequality yields
    		\begin{flalign*}
    			\mathbb{E}\big|X_{t\land\tau_{m,k}}-X_{t\land\tau_{m,k}}^N\big|^2
    			\leq&C_{m,k}^N + \mathbb{E}\Big[\int_0^{t\land\tau_{m,k}} C_k \big|X_s-X_s^N\big|^2 ds\Big]\\
    			=&C_{m,k}^N + C_k\int_0^t\mathbb{E}\Big[\big|X_s-X_s^N\big|^2\mathbf{1}_{\{s\leq \tau_{m,k}\}}\Big]ds\\
    			\leq&C_{m,k}^N+C_k\int_0^t\mathbb{E}\big|X_{s\land\tau_{m,k}}-X_{s\land\tau_{m,k}}^N\big|^2ds\\
    			\leq&C_{m,k}^N\exp(tC_k).
    		\end{flalign*}

            In the following proof, we regard the squared error $|X_t-X^N_t|^2$ as a sequence of random variables with respect to $N$ and prove it converges to 0 in the $L_1$ sense for any $t\in[0, T]$ by the Vitali convergence theorem. First, we prove $|X_t-X_t^N|^2$ converges to $0$ in probability for any $t\in[0, T]$. Specifically, for an arbitrary tolerance $\epsilon_P>0$, we have
			\begin{align*}
				P\big(\big|X_t-X_t^N\big|^2>\epsilon_P\big)
				=&P\big(\big\{\big|X_t-X_t^N\big|^2>\epsilon_P\big\} \cap \big\{\tau_{m,k}>t\big\}\big)\\
                &+P\big(\big\{\big|X_t-X_t^N\big|^2>\epsilon_P\big\}\cap\big\{\tau_{m,k}\leq t\big\}\big)\\
				\leq&P\big(\big|X_{t\land\tau_{m,k}}-X_{t\land\tau_{m,k}}^N\big|^2>\epsilon_P\big)+P(\tau_m\leq t)+P(\tau_k\leq t)\\
				\leq&\frac{1}{\epsilon_P}C_{m,k}^N\exp(tC_k)+P(\tau_m\leq T)+P(\tau_k\leq T),
			\end{align*}
			where the second inequality follows from Markov's inequality. 
			Given any arbitrarily small $\eta_P>0$, we can choose sufficiently large $m=m^*$ and $k=k^*$ such that
			\begin{equation*}
				P(\tau_{m^*}\leq T)+P(\tau_{k^*}\leq T)<\frac{\eta_P}{2}.
			\end{equation*}
			Since $\lim\limits_{N\to\infty}C_{m^*,k^*}^N=0$ and $C_{k^*}$ is constant for the fixed $k^*$, there exists a sufficiently large $N^*$ such that
			\begin{equation*}
				\frac{1}{\epsilon_P}C_{m^*,k^*}^N\exp(tC_{k^*})<\frac{\eta_P}{2},
			\end{equation*}
			for all $N\geq N^*$. 
			Consequently, for all $N\geq N^*$, we have $P\big(\big|X_t-X_t^N\big|^2>\epsilon_P\big)<\eta_P$, which implies $|X_t-X_t^N|^2\xrightarrow{P}0$ as $N\to\infty$ for all $t\in[0, T]$. 
			
			Moreover, the global linear growth conditions on $f$ and $g$, together with the definition of the admissible set $\mathcal{A}(\widehat{\mathbb{W}}^N)$, ensure that the fourth moment of $X_t^N$ is uniformly bounded with respect to $N$. 
			That is, $\sup\limits_{N}\mathbb{E}|X_t^N|^4<\infty$ and $\mathbb{E}|X_t|^4<\infty$. 
			Consequently, the fourth moment of the error also admits a uniform bound, i.e., $\sup\limits_N \mathbb{E}|X_t-X_t^N|^4<\infty$. 
			By the de la Vall\'ee Poussin theorem~\cite{dellacherie2011probabilities}, the family $\{|X_t-X_t^N|^2\}_{N\in\mathbb{N}}$ is uniformly integrable for any $t\in[0,T]$. 
			
			Finally, an application of the Vitali convergence theorem~\cite{williams1991probability} yields
			\begin{equation*}
				\lim\limits_{N\to\infty}\mathbb{E}\big|X_t-X_t^N\big|^2=0,\quad t\in[0,T]. 
			\end{equation*} 
    		}
    	\end{proof}

    \subsection{Proof of Theorem \ref{theorem-linear-option-price-gap}}\label{appendA3}

		\begin{proof}
			This proof is given by an application of Proposition \ref{proposition-X-gap}.
			\cpurple{
			\begin{flalign*}
				&\mathbb{E}\big|u(0,s_0)-u^N(0,s_0)\big|^2\leq\mathbb{E}\big|\Phi(X_T)-\Phi(X_T^N)\big|^2\leq L_{\Phi}^2\mathbb{E}\big|X_T-X_T^N\big|^2\leq L_{\Phi}^2\epsilon_T^N,
			\end{flalign*}
			}
			where $\epsilon_T^N$ comes from Proposition~\ref{proposition-X-gap}.
			In particular, as $N\to\infty$, $\epsilon_T^N\to0$ with bounded $L_{\Phi}$, which implies that the upper bound tends to zero.
		\end{proof}

    \subsection{Proof of Theorem \ref{theorem-nonlinear-rep}}\label{appendA4}
		\begin{proof}
			We first decompose the error as follows:
			\begin{flalign}
				&\mathbb{E}\big|\mathbf{v}_t-\mathcal{N}^{\mathbf{v}}_t(\widehat{\mathbb{W}}_t^N;\phi_{\mathbf{v}})\big|^2\nno\\ 
				\leq&3\left(\mathbb{E}\big|\mathbf{v}_t-\mathcal{L}(\widehat{\mathbb{W}}_t^{\infty})\big|^2 +\mathbb{E}\big|\mathcal{L}(\widehat{\mathbb{W}}_t^{\infty}) -\mathcal{L}(\widehat{\mathbb{W}}_t^N)\big|^2+\mathbb{E}\big|\mathcal{L} (\widehat{\mathbb{W}}_t^N)-\mathcal{N}^{\mathbf{v}}_t(\widehat{\mathbb{W}}_t^N;\phi_{\mathbf{v}})\big|^2\right),
				\label{eq-gap-nn}
			\end{flalign}
			where $\mathcal{L}$ denotes a linear functional. The above inequality comes from the inequality of arithmetic mean and quadratic mean, i.e., $(x_1+x_2+x_3)^2\leq 3 (x^2_1+x^2_2+x^2_3)$. According to the universal nonlinearity proposition (see, e.g.,~\cite{kidger2019deep}, Proposition A.6), for $(\mathcal{F}_t^W)$-progressive and continuous stochastic process $v$ and any $\epsilon^{\mathcal{N},\mathbf{v},1}_t>0$, there exists a linear functional $\mathcal{L}$ depending on a compact set, such that
			\begin{equation*}
				3\mathbb{E}\big|\mathbf{v}_t-\mathcal{L}(\widehat{\mathbb{W}}_t^{\infty})\big|^2\leq\epsilon^{\mathcal{N},\mathbf{v},1}_t.
			\end{equation*}
			The second term of~\eqref{eq-gap-nn} can be estimated by Theorem~\ref{theorem-linear-rep-gap} as follows
			\begin{equation*}
				3\mathbb{E}\big|\mathcal{L}(\widehat{\mathbb{W}}_t^{\infty})-\mathcal{L}(\widehat{\mathbb{W}}_t^N)\big|^2\leq \epsilon^{\mathbf{v},\widehat{\mathbb{W}},N}_t.
			\end{equation*}
			Next, applying the universality property (see, e.g.,~\cite{funahashi1993approximation}), for the linear functional and any $\epsilon^{\mathcal{N},\mathbf{v},2}_t>0$, there exists a neural network $\mathcal{N}^{\mathbf{v}}_t(\widehat{\mathbb{W}}_t^N;\phi_{\mathbf{v}})$, such that
			\begin{equation*}
				3\mathbb{E}\big|\mathcal{L}(\widehat{\mathbb{W}}_t^N)-\mathcal{N}^{\mathbf{v}}_t(\widehat{\mathbb{W}}_t^N;\phi_{\mathbf{v}})\big|^2\leq\epsilon^{\mathcal{N},\mathbf{v},2}_t.
			\end{equation*}
			Combining the above estimates, for a fixed neural network at time $t$, we conclude that
			\begin{equation*}
				\mathcal{E}^{\mathcal{N},\mathbf{v}}_t\leq\epsilon^{\mathcal{N},\mathbf{v},1}_t+\epsilon^{\mathcal{N},\mathbf{v},2}_t+\epsilon^{\mathbf{v},\widehat{\mathbb{W}},N}_t:=\epsilon^{\mathcal{N},\mathbf{v}}_t+\epsilon^{\mathbf{v},\widehat{\mathbb{W}},N}_t.
			\end{equation*}
		\end{proof}

\bibliographystyle{plain}
\bibliography{references}

\end{document}